%% file: main.tex
\documentclass[11pt]{article}
\usepackage{amsmath}
\usepackage{amsthm}
\usepackage{amssymb}
\usepackage{algorithm}
\usepackage{subfig}
\usepackage{color}
\usepackage[english]{babel}
\usepackage{graphicx}
\usepackage{wrapfig,epsfig}
\usepackage{epstopdf}
\usepackage{url}
\usepackage{graphicx}
\usepackage{color}
\usepackage{epstopdf}
\usepackage[noend]{algpseudocode}
\usepackage{scrextend}
\usepackage[T1]{fontenc}
\usepackage{bbm}
\usepackage{comment}
\usepackage{authblk}
\usepackage{multicol}
\usepackage{authblk}
\usepackage{url}
\usepackage{dsfont}
\usepackage{bbm}
\usepackage{bbm}
\usepackage{dsfont}

\usepackage{tikz}
\usepackage{xparse}
\usepackage{ifthen}
\usetikzlibrary{math}
\usepackage{hyperref}  
\hypersetup{colorlinks=true,citecolor=blue,linkcolor=blue} 
\usetikzlibrary{arrows}
\usepackage[margin=1in]{geometry}
\graphicspath{{./figs/}}

\author{
	Karl Bringmann\thanks{Saarland University and Max Planck Institute for Informatics, Saarland Informatics Campus, Saarbr\"ucken, Germany. \texttt{bringmann@cs.uni-saarland.de}}
	\quad
	Vasileios Nakos\thanks{Saarland University and Max Planck Institute for Informatics, Saarland Informatics Campus, Saarbr\"ucken, Germany. \texttt{vnakos@mpi-inf.mpg.de}}
}

\date{}

\title{Top-$k$-Convolution and \\the Quest for Near-Linear Output-Sensitive Subset Sum\thanks{This work is part of the project TIPEA that has received funding from the European Research Council (ERC) under the European Unions Horizon 2020 research and innovation programme (grant agreement No.~850979).}}

\newtheorem{theorem}{Theorem}[section]
\newtheorem{lemma}[theorem]{Lemma}
\newtheorem{definition}[theorem]{Definition}

\newtheorem{corollary}[theorem]{Corollary}

\newtheorem{observation}[theorem]{Observation}

\newtheorem{claim}[theorem]{Claim}

\newtheorem{question}[theorem]{Question}

\makeatletter
\newcommand{\ostar}{\mathbin{\mathpalette\make@circled\star}}
\newcommand{\make@circled}[2]{%
  \ooalign{$\m@th#1\smallbigcirc{#1}$\cr\hidewidth$\m@th#1#2$\hidewidth\cr}%
}
\newcommand{\smallbigcirc}[1]{%
  \vcenter{\hbox{\scalebox{0.77778}{$\m@th#1\bigcirc$}}}%
}
\makeatother

\newcommand{\conv}{\star}
\newcommand{\convBool}{\ostar}

\newcommand{\wt}{\widetilde}

\newcommand{\eps}{\epsilon}

\renewcommand{\i}{\mathbf{i}}

\renewcommand{\varepsilon}{\epsilon}
\renewcommand{\tilde}{\wt}

\renewcommand{\eps}{\epsilon}

\newcommand{\out}{\mathrm{out}}

\DeclareMathOperator{\supp}{supp}

\DeclareMathOperator{\poly}{poly}

\DeclareMathOperator{\polylog}{polylog}

\newcommand{\Oh}{O}
\newcommand{\tOh}{\widetilde{O}}

\newcommand*{\xMin}{1}%
\newcommand*{\xMax}{9}%
\newcommand*{\yMin}{1}%
\newcommand*{\yMax}{9}%

\newcommand{\karl}[1]{\textbf{\color{red}[Karl: #1]}}

\makeatletter
\newcommand*{\RN}[1]{\expandafter\@slowromancap\romannumeral #1@}
\makeatother

\usepackage{lineno}

\input{libtheorems.tex}

\begin{document}

\begin{titlepage}
  \maketitle
  \input{abstract}
  \thispagestyle{plain}
\end{titlepage}
\newpage
\tableofcontents
\thispagestyle{plain}
\newpage

\setcounter{page}{1}

\input{introduction}

\input{results}

\input{topkconvolution}

\input{sparse_convolutions}

\input{reduction}

\input{conclusion}

\addcontentsline{toc}{section}{References}
\bibliographystyle{plain}
\bibliography{ref}






\end{document}

%% file: libtheorems.tex
\usepackage{etoolbox}

\makeatletter

\newcommand{\define}[4][ignore]{%
  \ifstrequal{#1}{ignore}{}{
  \@namedef{thmtitle@#2}{#1}}%
  \@namedef{thm@#2}{#4}%
  \@namedef{thmtypen@#2}{lemma}%
  \newtheorem{thmtype@#2}[theorem]{#3}%
  \newtheorem*{thmtypealt@#2}{#3~\ref{#2}}%
}

\newcommand{\state}[1]{%
  \@namedef{curthm}{#1}
  \@ifundefined{thmtitle@#1}{
  \begin{thmtype@#1}
    }{
  \begin{thmtype@#1}[\@nameuse{thmtitle@#1}]
  }
    \label{#1}
    \@nameuse{thm@#1}
  \end{thmtype@#1}
  \@ifundefined{thmdone@#1}{
  \@namedef{thmdone@#1}{stated}%
  }{}
}

\newcommand{\restate}[1]{%
  \@namedef{curthm}{#1}
  \@ifundefined{thmtitle@#1}{
    \begin{thmtypealt@#1}
    }{
  \begin{thmtypealt@#1}[\@nameuse{thmtitle@#1}]
  }
    \@nameuse{thm@#1}
  \end{thmtypealt@#1}
  \@ifundefined{thmdone@#1}{
  \@namedef{thmdone@#1}{stated}%
  }{}
}

\newcommand{\thmlabel}[1]{
  \@ifundefined{thmdone@\@nameuse{curthm}}{\label{#1}
    }{\tag*{\eqref{#1}}}
}
\makeatother

%% file: abstract.tex
\begin{abstract}

In the classical \textsc{SubsetSum} problem we are given a set $X$ and a target $t$, and the task is to decide whether there exists a subset of $X$ which sums to $t$. A recent line of research has resulted in $\tilde{O}(t)$-time algorithms, which are (near-)optimal under popular complexity-theoretic assumptions. On the other hand, the standard dynamic programming algorithm runs in time $O(n \cdot |\mathcal{S}(X,t)|)$, where $\mathcal{S}(X,t)$ is the set of all subset sums of $X$ that are smaller than~$t$. Furthermore, all known pseudopolynomial algorithms actually solve a stronger task, since they actually compute the whole set $\mathcal{S}(X,t)$.

As the aforementioned two running times are incomparable, in this paper we ask whether one can achieve the best of both worlds: running time $\tilde{O}(|\mathcal{S}(X,t)|)$. In particular, we ask whether $\mathcal{S}(X,t)$ can be computed in near-linear time in the output-size. 
Using a diverse toolkit containing techniques such as color coding, sparse recovery, and sumset estimates, we make considerable progress towards this question and design an algorithm running in time $\tilde{O}(|\mathcal{S}(X,t)|^{4/3})$. 

Central to our approach is the study of \emph{top-$k$-convolution}, a natural problem of independent interest: given sparse polynomials with non-negative coefficients, compute the lowest $k$ non-zero monomials of their product.
We design an algorithm running in time $\tOh(k^{4/3})$, by a combination of sparse convolution and sumset estimates considered in Additive Combinatorics. Moreover, we provide evidence that going beyond some of the barriers we have faced requires either an algorithmic breakthrough or possibly new techniques from Additive Combinatorics on how to pass from information on restricted sumsets to information on unrestricted sumsets.

\end{abstract}

%% file: introduction.tex

\section{Introduction}

\subsection{Subset Sum}
\label{sec:introsubsetsum}

\textsc{SubsetSum} is a fundamental problem at the intersection of computer science, mathematical optimization, and operations research. 
In this problem, given a set $X$ of $n$ integers and a target~$t$, the task is to decide whether there exists a subset of $X$ that sums to $t$. The problem belongs to Karp's initial list of \textsc{NP}-complete problems~\cite{Karp72}, and it has given rise to a plethora of algorithmic techniques, see, e.g., the monographs~\cite{KPP04book,MT90book}. 
Apart from being a cornerstone in algorithm design, \textsc{SubsetSum} draws its importance from being a special case of many other problems, like \textsc{Knapsack} or \textsc{Integer Programming}. It has also played a role in cryptography, as Merkle and Hellman~\cite{MH78} based their cryptosystem on this problem, see also~\cite{Sha84,BO88,CR88,Odl90,IN96}. 

Several classic algorithms for \textsc{SubsetSum} are typically taught in undergraduate courses, including the meet-in-the-middle algorithm running in time $O(2^{n/2})$~\cite{HS74} and Bellman's dynamic programming algorithm running in pseudopolynomial time $O(n \cdot t)$~\cite{bellman1957dynamic}.

Surprisingly, after decades of research, major algorithmic advances were still discovered in the last 10 years, e.g.,~\cite{OB11,LMS11a,DKS12,Austrin13,GS15,AKKN15,AKKN16,LWWW16,
BGNV17,Ned17,Bring17,KoiliarisX17,JH19,ABHDD19,ABJTW19}. 
Among these developments, the most relevant for this paper are improvements over Bellman's $O(n \cdot t)$ algorithm: Koiliaris and Xu~\cite{KoiliarisX17} designed a deterministic algorithm running in time\footnote{By $\tilde{O}(T)$ we hide factors of the form $\polylog(T)$ as well as factors $\polylog(u)$, where $u$ is the universe size, and $\polylog(t)$, where $t$ is the target.} $\tilde{O}( \mathrm{min}\{ \sqrt{n} \cdot t, t^{4/3} \})$, and Bringmann~\cite{Bring17} devised a randomized algorithm running in time $\tilde{O}(t)$ (which was improved in terms of log factors in~\cite{JH19}). 
The running time $\tilde{O}(t)$ of the randomized algorithms is optimal 
under the Strong Exponential Time Hypothesis~\cite{ABHDD19} as well as under the \textsc{SetCover} Hypothesis~\cite{Cygan+16}. 

Thus, research on pseudopolynomial algorithms for \textsc{SubsetSum} with respect to parameter $t$ is more or less finished. However, it remains to study whether the recent improvements generalize to other parameters as well as to variants of \textsc{SubsetSum}. For instance, this has been done for the \textsc{ModularSubsetSum} problem in~\cite{ABJTW19}. 
In this paper, we start from the observation that Bellman's classic dynamic programming algorithm can be implemented to run in time $O(n \cdot |\mathcal{S}(X,t)|)$, where $\mathcal{S}(X,t)$ is the set of all subset sums of $X$ below $t$. 
Since $|\mathcal{S}(X,t)|$ can be much smaller than~$t$, so far the running times $\tOh(t)$ and $O( n \cdot |\mathcal{S}(X,t)|)$ are incomparable. Thus, despite the running time $\tOh(t)$ being matched by a conditional lower bound, in situations where $|\mathcal{S}(X,t)|$ is small Bellman's algorithm can outperform the recent improved algorithms. To obtain the best of both worlds, it would thus be desirable to consider $|\mathcal{S}(X,t)|$, rather than $t$, as the parameter to measure the computational complexity of the problem, and to similarly shave off the factor~$n$ from the running time of Bellman's algorithm. In particular, since all previous pseudopolynomial algorithms for \textsc{SubsetSum} produce all attainable subset sums smaller than $t$, a natural question is whether one can design a \emph{near-linear output-sensitive} algorithm.

\begin{question} \label{question:Q1}
Is there an algorithm that computes $S(X,t)$ in time $\tOh(|\mathcal{S}(X,t)|)$?
\end{question}


Our work struggles to make progress towards understanding \textsc{SubsetSum} under this new computational perspective, and it lead us to study a new type of sparse convolution problem.

\subsection[Top-$k$-Convolution]{Top-\boldmath$k$-Convolution}

Convolution and Boolean convolution are fundamental computational primitives that frequently arise in algorithm design, e.g., when combining solutions of two subproblems.
The \emph{Boolean convolution} $f \convBool g$ of vectors $f,g \in \{0,1\}^u$ is the vector with entries $(f \convBool g)_k = \bigvee_{0 \le i \le k} f_i \wedge g_{k-i}$ for $0 \le k < 2u$. This arises when we split a problem into two subproblems, so that the whole problem has a solution of size $k$ if and only if for some $i$ the left subproblem has a solution of size $i$ and the right subproblem has a solution of size $k-i$. Moreover, Boolean convolution is equivalent to \emph{sumset computation}, where we are given sets $A,B \subseteq \{0,1,\ldots,u-1\}$ and the task is to compute $A+B$, the set of all sums $a+b$ with $a \in A,\, b \in B$. 

The \emph{convolution} $f \conv g$ of vectors $f,g \in \mathbb{R}^u$ is the vector with entries $(f \conv g)_k = \sum_{i=0}^k f_i \cdot g_{k-i}$ for $0 \le k < 2u$. When we split a problem into two subproblems, and $f_i$ and $g_i$ count the number of size-$i$ solutions of the left and right subproblem, then $(f \conv g)_k$ counts the number of size-$k$ solutions of the whole problem. Moreover, convolution is equivalent to polynomial multiplication, where we are given the coefficients of two polynomials and want to compute the coefficients of their product.

Boolean convolution can be solved via convolution, and convolution can be solved in time $\Oh(u \log u)$ using Fast Fourier Transform (FFT). 
However, when the input vectors are sparse, we can ask for algorithms that compute convolutions much faster than performing FFT. 
Specifically, the ultimate goal is an algorithm running in near-linear output-sensitive time, i.e., in near-linear time in terms of the number of non-zero entries of $f \conv g$.
This practically and theoretically relevant problem is called sparse convolution (or sparse polynomial multiplication) and has been studied, e.g., in~\cite{CH02,roche2008adaptive,monagan2009parallel,van2012complexity,AR15,CL15,roche2018can,N19,BriNak19}.
The ultimate goal of near-linear output-sensitive time has been achieved by Cole and Hariharan~\cite{CH02} for non-negative vectors by a Las Vegas algorithm, in~\cite{N19} for general vectors by a Monte Carlo algorithm, and in~\cite{BriNak19} for non-negative vectors by an almost linear time deterministic algorithm. 

\medskip
In this paper we study a natural variant that we call \emph{Top-$k$-Convolution}: Given vectors $f,g \in \mathbb{R}^u$, compute the $k$ lowest non-zero entries of $f \conv g$. Formally, denoting by $i_1 < i_2 < \ldots < i_\ell$ all indices of non-zero entries of $f \conv g$, our goal is to compute the pairs $(i_1, (f \conv g)_{i_1}), \ldots, (i_k, (f \conv g)_{i_k})$. \emph{Boolean Top-$k$-Convolution} is defined analogously, i.e., the task is to compute the lowest $k$ indices $i_1<\ldots<i_k$ of 1-entries of $f \convBool g$. 

``Top-$k$'' problems, that ask for the $k$ best solutions, are well motivated from a practical perspective, e.g., for displaying search results.
Note that in the setting where we split a problem into two subproblems, Boolean Top-$k$-Convolution asks for the $k$ smallest solution sizes. In the polynomial multiplication setting, Top-$k$-Convolution asks for the $k$ lowest-degree monomials in the product of two given sparse polynomials. The problem is equivalent if we replace ``lowest'' by ``highest'' (by reversing $f$ and $g$), and thus Top-$k$-Convolution is also equivalent to computing the $k$ highest-degree monomials in the product of two given sparse polynomials. 
As an additional major application, we present a connection to \textsc{SubsetSum} in this paper.


Therefore, sparse Top-$k$-Convolution is a well-motivated problem with several applications, and it is surprising that to the best of our knowledge it has not been explicitly studied before. 
Throughout the paper we will assume that $f,g$ are non-negative vectors, since this is satisfied in many applications and the possible cancelations resulting from negative entries make the problem considerably harder (indeed, sparse convolution with negative entries has been solved much later~\cite{N19} than on non-negative vectors~\cite{CH02}).
Note that Top-$k$-Convolution on non-negative vectors can be solved naively in time $\Oh(k^2)$: The lowest $k$ non-zeros of $f \conv g$ are among the combinations of the lowest $k$ non-zeros of $f$ and the lowest $k$ non-zeros of $g$. This allows us to assume that $f$ and $g$ are $k$-sparse, and we can compute $f \conv g$ naively in time $\Oh(k^2)$ to obtain its lowest $k$ non-zeros.

\subsection{Our Contribution}
We initiate the study of \textsc{SubsetSum} with respect to the parameter $|\mathcal{S}(X,t)|$, making considerable progress in this direction. We show an ``output-sensitivity preserving'' reduction from \textsc{SubsetSum} to Top-$k$-Convolution, such that if the latter can be solved in near-linear output-sensitive time so can the former. Then, we investigate upper and lower bounds for Top-$k$-Convolution on non-negative vectors (as well as related restricted variants of sumset computation and sparse convolution). 
In particular, we present a randomized $\tOh(k^{4/3})$-time algorithm for Top-$k$-Convolution, resulting in time $\tOh(|\mathcal{S}(X,t)|^{4/3})$ for \textsc{SubsetSum}.
Our algorithms fall into a natural class that we call ``rectangle covering algorithms'', for which we show that they cannot yield near-linear time.
Our technical machinery draws from a wide range of techniques such as color coding, sparse convolutions, and sumset estimates from Additive Combinatorics.

%% file: results.tex
\section{Results and Techniques}

\subsection{Preliminaries}

We write $\mathbb{N} = \{0,1,\ldots\}$ and $[n] = \{0,1,\ldots,n\}$ for any $n \in \mathbb{N}$.
For sets $A,B \subseteq \mathbb{N}$ define $A+B =  \{ a+b \mid a \in A,\, b \in B \}$. We also define $\max(A) = \max_{x \in A} x$, and similarly $\mathrm{min}(A)$. 
For any set $X \subseteq \mathbb{N}$ we define $\Sigma(X) = \sum_{x \in X} x$. For $t \in \mathbb{N}$ we define the set of all subset sums of $X$ below $t$ as
\[	\mathcal{S}(X,t) = \{ \Sigma(Y) \mid Y \subseteq X,\, \Sigma(Y) \le t \}.	\]
For integers $a,b$ we write $[a,b] = \{a, \ldots, b\}, (a,b) = \{a+1,\ldots, b-1\}$, and similarly for $(a,b]$ and $[a,b)$. 

For a vector $f \in \mathbb{R}^d$ we let $\|f\|_0$ be its number of non-zero entries. For a set $S \subseteq [d]$ we denote by $f_S$ the vector that is zeroed out outside of $S$, i.e., $(f_S)_i = f_i$ for $i \in S$ and $(f_S)_i = 0$ otherwise. We shall use $0$-indexed vectors throughout the paper. We also use a non-standard notation of $\tilde{O}(T)$ throughout the paper that hides factors of the form $\polylog(T)$ as well as $\polylog(u)$, where $u$ is the universe size, or $\polylog(t)$, where $t$ is the target.

We shall need the following result by Cole and Hariharan and an immediate corollary. 

\begin{theorem}[\cite{CH02}] \label{thm:cole_hariharan}
Given non-negative vectors $f,g$ of length $d$, we can compute $f \conv g$ in expected time $O\left( \|f \conv g\|_0 \cdot \log^2 d \right)$.
\end{theorem}

\begin{theorem} \label{thm:sparse_convolution}
Given $A,B \subseteq [u]$, we can compute $A+B$ in expected time $O\left( |A+B| \cdot \log^2 u \right)$.
\end{theorem}
\begin{proof}
  Let $f$ be the indicator vector of $A$, $g$ be the indicator vector of $B, d = 2u+1$ and use Theorem~\ref{thm:cole_hariharan}.
\end{proof}

\subsubsection{Problem Definitions}
Our work is concerned with the following problems.

\begin{definition}[\textsc{Subset Sum}]
Given a set $X \subseteq \mathbb{N}$ and a number $t$, compute $\mathcal{S}(X,t) = \{ \Sigma(Y) \mid Y \subseteq X,\, \Sigma(Y) \le t \}$. We measure the running time in terms of $|\mathcal{S}(X,t)|$. We write $n = |X|$.
\end{definition}

\begin{definition}[Top-$k$-Convolution]
Given two vectors $f,g \in \mathbb{R}^u$ and a parameter $k$, compute the first $k$ non-zero entries of $f \conv g$. In other words, find $(i_1, (f\conv g)_{i_1}), (i_2, (f\conv g)_{i_2}), \ldots (i_k, (f\conv g)_{i_k})$ where $i_1 < i_2 < \ldots < i_k$ are the smallest $k$ indices on which $ f \conv g$ is non-zero. We measure the running time in terms of $k$. We write $n = \|f\|_0$ and $m = \|g\|_0$.
\end{definition}

Top-$k$-convolution is equivalent to the following problem, as we will show in Lemma~\ref{lem:equ_topk_prefix}.

\begin{definition}[Prefix-Restricted Convolution]
Given a positive integer~$u$ and two vectors $f,g \in \mathbb{R}^u$, compute $(f \conv g)_{[u]}$, i.e., compute (a sparse representation of) the first $u$ entries of $f \conv g$. We measure the running time in terms of the output-size $\out = \| (f \conv g)_{[u]} \|_0$. We write $n = \|f\|_0$ and $m = \|g\|_0$.
\end{definition}




The following problem is a Boolean version of prefix-restriced convolution.

\begin{definition}[Prefix-restricted Sumset Computation]
Given $u \in \mathbb{N}$ and $A,B \subseteq [u]$, compute $(A+B)\cap [u]$. We measure the running time in terms of the output-size $\out = |(A+B)\cap [u]|$. We write $n = |A|$ and $m = |B|$.
\end{definition}




\subsubsection{Covering of Prefix-Restricted Sumsets}

For a set $A \subseteq \mathbb{N}$ of size $n$, we implicitly assume that $A$ is sorted and we denote its elements in sorted order as $A_1 < A_2 < \ldots < A_n$. Moreover, for $I \subseteq \{1,\ldots,n\}$ we write $A_I$ for $\{A_i \mid i \in I\}$. 

We define the notion of a covering of a restricted sumset $(A,B,[u])$. Intuitively, we want to cover $(A+B)\cap [u]$ by sumsets of the form $A_I+B_J$, that is, we want to have $(A+B)\cap [u] \subseteq \bigcup_{(I,J) \in \mathcal{C}} A_I + B_J$.

\begin{definition}
Let $u \in \mathbb{N}$ and $A,B \subseteq [u]$ with $n=|A|,\, m=|B|$.
A \textbf{covering} of $(A,B,[u])$ is a family $\cal C$ such that:
\begin{enumerate}
\item $\mathcal{C}$ consists of pairs $(I,J)$ where $I \subseteq \{1,\ldots,n\}$ and $J \subseteq \{1,\ldots,m\}$.
\item For any $1 \le i \le n,\, 1 \le j \le m$ with $A_i + B_j \in [u]$ there exists $(I,J) \in \mathcal{C}$ with $(i,j) \in I \times J$.
\end{enumerate}
We call a covering $\cal C$ \textbf{unique} if the pair in property 2.\ is unique, i.e., for any $1 \le i \le n,\, 1 \le j \le m$ with $A_i + B_j \in [u]$ there exists a \emph{unique} pair $(I,J) \in \mathcal{C}$ such that $(i,j) \in I \times J$.

We call a covering $\cal C$ a \textbf{rectangle covering} if for any $(I,J) \in \cal C$ the sets $I$ and $J$ are intervals, meaning that they consists of contiguous elements in the sorted order of $A$ and $B$, respectively.

The \textbf{cost} of a covering $\cal C$ is
	\[	\sum_{ (I,J) \in \mathcal{C} } \left| A_I + B_J \right|. 	\]
\end{definition}

This notion is useful because of the following fact.

\begin{observation} \label{obs:sumset_from_covering}
  Given a covering $\cal C$ of $(A,B,[u])$ of cost $c$, we can compute the prefix-restricted sumset $(A+B) \cap [u]$ in expected time $\tOh(c)$.
\end{observation}
\begin{proof}
  Using output-sensitive sumset computation (Theorem~\ref{thm:sparse_convolution}), we can compute $A_I + B_J$ in expected time $\Oh(|A_I+B_J| \log^2 u)$. Thus, we can compute $R := \bigcup_{(I,J) \in \mathcal{C}} A_I + B_J$ in time $\Oh(c \log^2 u)$. By the covering property, we can simply return $R \cap [u] = (A+B) \cap [u]$.  
\end{proof}


We refer to an algorithm making use of Observation~\ref{obs:sumset_from_covering} as a \emph{covering algorithm}. 
We call it a \emph{unique-rectangle-covering algorithm} if the used covering is a unique rectangle covering.
This is a natural class of algorithms, as we also explain in Section \ref{sec:tech_natural}. 
All algorithms presented in this paper are unique-rectangle-covering algorithms.

\subsection{Formal Statement of Results}
\label{sec:results}

\subsubsection{Main Results}

As the technical core of our paper, we present an efficient construction of low-cost coverings.

\begin{theorem}[Covering Construction, Section~\ref{sec:technical_core}] \label{thm:technical_core}
Given $A,B \subseteq [u]$, in expected time $\tOh(\out^{4/3})$ we can compute a unique rectangle covering of $(A,B,[u])$ of cost $\tOh(\out^{4/3})$, where $\out = |(A+B) \cap [u]|$.
\end{theorem}

By Observation~\ref{obs:sumset_from_covering}, this yields an $\tOh(\out^{4/3})$ Las Vegas algorithm for prefix-restricted sumset computation. Via simple reductions, we obtain similar algorithms for our convolution problems.

\begin{corollary}[Top-$k$-Convolution and Related Problems, Section~\ref{sec:top_k_conv_reductions}] \label{cor:top_k_and_related_algos}
  Top-$k$-convolution on non-negative vectors can be solved in expected time $\tOh(k^{4/3})$. Prefix-restricted sumset computation and prefix-restricted convolution on non-negative vectors can be solved in expected time $\tOh(\out^{4/3})$. 
\end{corollary}

By carefully adapting a recent pseudopolynomial $\tOh(n+t)$-time \textsc{SubsetSum} algorithm~\cite{Bring17} to use prefix-restricted sumset computation as a subroutine, we obtain our main result for \textsc{SubsetSum}.
This can be seen as a reduction from \textsc{SubsetSum} to prefix-restricted sumset computation.

\begin{theorem}[\textsc{SubsetSum}, Section~\ref{sec:subset_sum}] \label{thm:subsetsum_result}
  Given $X \subseteq \mathbb{N}$ and $t \in \mathbb{N}$, we can compute the set $\mathcal{S}(X,t)$ in time $\tOh(|\mathcal{S}(X,t)|^{4/3})$ with high probability.
\end{theorem}

Since all of these results depend on our technical core (Theorem~\ref{thm:technical_core}), we also study limitations of rectangle-covering algorithms. The following result shows that our approach of using rectangle coverings to solve top-$k$-convolution and \textsc{SubsetSum} cannot achieve near-linear output-sensitive algorithms. 

\begin{theorem}[Lower Bound on Rectangle Coverings, Section~\ref{sec:covering_lb}] \label{thm:covering_lb}
There exists an infinite sequence of tuples $(A,B,[u])$, with $u \in \mathbb{N}$ and $A,B \subseteq [u]$, such that any rectangle covering of $(A,B,[u])$ has cost $\Omega ( \out^{1.047} )$, where $\out = |(A+B) \cap [u]|$.
\end{theorem}

We remark that this result crucially uses \emph{rectangle} coverings; we do not rule out the existence of non-rectangle coverings of near-linear cost.

In the remainder of Section~\ref{sec:results} we present additional related results.

\subsubsection{Additional Results: Relaxing the Upper Bound}

In our results so far we have relaxed the ultimate goal of algorithms running in time $\tOh(\out)$ to time $\tOh(\out^{4/3})$. We can alternatively relax the goal by measuring the running time in terms of a slightly larger upper bound. Specifically, for prefix-restricted sumset computation so far we measured time in terms of the output-size $\out = |(A+B) \cap [u]|$. Now we will measure the running time in terms of $|(A+B) \cap [(1+\zeta)u]|$ for some $\zeta > 0$, while the task still is to compute the set $(A+B) \cap [u]$. 
Since one could expect that for ``realistic'' instances $(A,B,u)$ there are not many sums in the boundary region $(A+B) \cap [u,(1+\zeta)u]$, algorithms that perform well with respect to this new measure might perform well in practice.
We show the following.

\begin{theorem}[Relaxed Upper Bound, Section~\ref{sec:prefix_zeta}]\label{thm:add_relaxed_upper_bound}
  Prefix-restricted sumset computation can be solved in expected time $\tOh( \zeta^{-1} \cdot |(A+B) \cap [(1+\zeta)u]| + \zeta^{-2})$. 
  Prefix-restricted convolution on non-negative vectors can be solved in expected time $\tOh( \zeta^{-1} \cdot \|(f \conv g)_{[(1+\zeta)u]}\|_0 + \zeta^{-2} )$. 
  \textsc{SubsetSum} can be solved in time $\tOh(\zeta^{-1} \cdot |\mathcal{S}(X,(1+\zeta)t)| + \zeta^{-2})$ with high probability.
\end{theorem}

Note that $|\mathcal{S}(X,(1+\zeta)t)| \le (1+\zeta)t$, and thus plugging in any constant $\zeta > 0$ yields an algorithm running in time $\tOh(t)$, which is as good as~\cite{Bring17,JH19}.

\subsubsection{Additional Results: Interval-Restricted Sumset and Convolution}

So far we studied the problem of computing the first $u$ output values of sumset computation and convolution. More generally, we could ask to compute all output values at given positions $P \subseteq \mathbb{N}$. This is a well-motivated generalization, as it corresponds to computing specific coefficients of the product of two polynomials, which comes up in counting problems. 

However, here we show that even when the desired positions $P$ form an \emph{interval} then near-linear output-sensitive algorithms are unlikely to exist. Specifically, we study the following problems.

\begin{definition}[Interval-Restricted Sumset Computation]
Given positive integers $\ell \le u$ and $A,B \subseteq [u]$, compute $(A+B)\cap [\ell,u]$. We denote the output-size by $\out = |(A+B)\cap [\ell,u]|$. We write $n=|A|$ and $m=|B|$.
\end{definition}

\begin{definition}[Interval-Restricted Convolution]
Given positive integers $\ell \le u$ and vectors $f,g \in \mathbb{R}^u$, compute $(f \conv g)_{[\ell,u]}$, i.e., compute (a sparse representation of) the entries of $(f \conv g)_i$ for $\ell \le i \le u$. We denote the output-size by $\out = \| (f \conv g)_{[\ell,u]} \|_0$. We write $n = \|f\|_0$ and $m = \|g\|_0$.
\end{definition}

We design algorithms for these problems with the following running time. This uses a natural generalization of coverings, where we simply replace the set $[u]$ by $[\ell,u]$.

\begin{theorem}[Algorithm for Interval-Restricted] \label{thm:add_interval_algo}
  Interval-restricted sumset computation and interval-restricted non-negative convolution can be solved in expected time $\tOh(n + m + \sqrt{n m\, \out})$. 
\end{theorem}

We present conditional lower bounds based on two classical problems: Boolean matrix multiplication and sliding window Hamming distance. In Boolean matrix multiplication, we are asked to multiply two Boolean $n \times n$-matrices. By a reduction to multiplying real-valued matrices, Boolean matrix multiplication can be solved in time $\Oh(n^{\omega})$, where $\omega \le 2.373$ is the exponent of fast matrix multiplication~\cite{vw12,g14}. 
Limitations of fast matrix multiplication have been recently studied, and it was shown that known techniques cannot facilitate $\omega < 2 + \frac{1}{6}$~\cite{avw18a,avw18b,a19}. In particular, $\omega > 2$ is a barrier that known techniques cannot overcome, also for Boolean matrix multiplication. 

In the sliding window Hamming distance problem, we are given a text of length $2n$ and a pattern of length $n$, both over a (possibly large) alphabet $\Sigma$, and want to find the Hamming distance between the pattern and every length-$n$ substring of the text. This problem admits a $\tilde{O}(n^{3/2})$-time algorithm, and it is a major open question in string algorithms whether a faster algorithm exists~\cite{A87,gu17}.

\begin{theorem}[Hardness for Interval-Restricted] \label{thm:interval_hardness}
  Let $\delta > 0$.
  If interval-restricted sumset computation is in time $\tOh((n + m + \out)^{\omega/2 - \delta})$, then Boolean matrix multiplication is in time $\tOh(n^{\omega - 2 \delta})$. 
  
  If interval-restricted convolution on non-negative vectors is in time $\tOh(n + m + (n m\, \out)^{1/2 - \delta})$, then sliding window Hamming distance is in time $\tOh(n^{3/2 - 3 \delta})$. 

\end{theorem}

\subsection{Overview of Techniques} 

\subsubsection{Reduction from Subset Sum to Prefix-Restricted Sumset Computation}

To reduce \textsc{SubsetSum} to prefix-restricted sumset computation, our starting point is the pseudopolynomial $\tOh(t)$-time algorithm of~\cite{Bring17}, which uses a combination of color coding and FFT. We observe that in this algorithm all uses of FFT in fact compute prefix-restricted sumsets, so we can replace FFT by an algorithm for prefix-restricted sumset computation. This directly yields a reduction from \textsc{SubsetSum} to prefix-restricted sumset computation, but it does not directly give the desired guarantees. Our goal is to make this reduction ``output-sensitivity preserving'', that is, if we can compute prefix-restricted sumsets in near-linear output-sensitive time then the reduction yields an algorithm for \textsc{SubsetSum} that runs in near-linear output-sensitive time. 

To this end, we need to take a closer look at the algorithm in~\cite{Bring17}.
Given $(X,t)$, this algorithm splits $X$ into the set $X^{(L)}$ of numbers that are larger than $t / \polylog(n)$ and the remaining small numbers $X^{(S)}$. Any subset summing to at most $t$ uses at most $\polylog(n)$ large numbers, which enables a color-coding-based algorithm in~\cite{Bring17}. We observe that, without any problems, replacing the usage of FFT in that part of the algorithm by near-linear output-sensitive prefix-restricted sumset computation yields a near-linear output-sensitive algorithm for handling the large numbers~$X^{(L)}$. 
The small numbers $X^{(S)}$ are then further split at random into two subsets $X^{(1)},X^{(2)}$. By concentration inequalities, this splits the subset $Y$ into two subsets $Y^{(i)} = Y \cap X^{(i)}$ satisfying with high probability $\Sigma(Y^{(i)}) \le (1+\eps) t/2$. 
We can therefore find $\Sigma(Y^{(1)})$ and $\Sigma(Y^{(2)})$ by recursive calls on $(X^{(1)}, (1+\eps) t/2)$ and $(X^{(2)}, (1+\eps) t/2)$. Since we recurse on two subproblems and the target bound~$t$ is roughly split in half, one can argue that the total running time is $\tOh(t)$~\cite{Bring17}.


In this paper, we instead measure the running time in terms of $|\mathcal{S}(X,t)|$. Hence, it does not suffice that the target bound $t$ is roughly split in half; we also need that the measure $|\mathcal{S}(X,t)|$ is roughly split in half. Indeed, if we can show this, then we again obtain two subproblems and the measure $|\mathcal{S}(X,t)|$ is roughly split in half, so we can argue that the running time is $\tOh(|\mathcal{S}(X,t)|)$.
This requires novel insights into the set of all subset sums, to obtain inequalities relating the number of subset sums of $X$ to that of its subsets. 
Specifically, we show that $X^{(1)}$ and $X^{(2)}$ satisfy 
\[ |\mathcal{S}(X^{(1)},(1+\eps) t/2)| + |\mathcal{S}(X^{(2)},(1+\eps) t/2)| \le (1+\Oh(\eps)) |\mathcal{S}(X,t)| + \Oh(1). \]
This inequality only holds if every element of $X^{(1)},X^{(2)}$ is sufficiently smaller than the target $t$. Indeed, in the proof of this inequality we use that all numbers in $X^{(1)},X^{(2)}$ are small, which justifies why we removed the large numbers $X^{(L)}$.

\subsubsection{Restricted Sumset Computation}

We now give an overview of our techniques for restricted sumset computation, and discuss the relation to classical results from Additive Combinatorics.

Consider the prefix-restricted sumset computation of $(A+B) \cap [u]$. Note that we can assume $A,B \subseteq [u]$, since larger numbers cannot form sums of at most $u$.
We also always implicitly assume that for every number $a \in A$ there is a number $b \in B$ with $a + b \le u$, since otherwise we can remove $a$ from $A$ without changing $(A+B) \cap [u]$. In other words, we assume $\max(A) + \min(B) \le u$, and symmetrically $\min(A) + \max(B) \le u$, which can be ensured after a $\Oh(|A|+|B|)$-time preprocessing.
 
\paragraph{Standard sumset computation does not suffice.}
Standard sparse convolution (Theorem \ref{thm:sparse_convolution}) allows us to compute $A+B$ in near-linear output-sensitive time $\tOh(|A+B|)$. 
In particular, since $A+B \subseteq [2u]$, we can compute $(A+B) \cap [2u]$ in near-linear output-sensitive time. 
Although this might seem close at first glance, careful inspection reveals that it can be much larger than the goal $|(A+B)\cap [u]|$. Indeed, we can construct sets $A,B \subseteq [u]$ with $|A|=|B|=n$ satisfying $|(A+B)\cap [2u]| = \Omega(n^2)$ and $|(A+B)\cap[u]| = \Oh(n)$. Hence, in general we cannot afford to compute $(A+B)\cap [2u]$ when our goal is to compute $(A+B)\cap[u]$ in near-linear output-sensitive time. 

For this construction, assume that $u$ is a sufficiently large even integer and let 
\begin{align*}
  &X := [n] = \{0,1,2,\ldots,n\}, && Y := n [n] = \{0,n,2n,\ldots,n^2\}, \\
  &A := \{0\} \cup (\{\tfrac u2\} + X), && B := \{0\} \cup (\{\tfrac u2\} + Y).
\end{align*}
One can check that $\max(A) + \min(B),\, \min(A) + \max(B) \le \frac u2 + n^2 \le u$. Since for any $x \in X,\, y \in Y$ the sum $(\frac u2 + x) + (\frac u2 + y)$ is larger than $u$, we obtain 
$|(A+B) \cap [u]| \le |\{0\}| + |X| + |Y| = 2n+1$.
Moreover, one can check that $X+Y = [n^2+n]$, and thus $|(A+B) \cap [2u]| = |A+B| \ge n^2$. 

\paragraph{A natural class of algorithms.} \label{sec:tech_natural} 
We investigate a natural and safe algorithmic approach: To compute $(A+B)\cap [u]$, we make sure that every sum $A_i + B_j \leq u$ is computed, where $1 \le i \le n$, $1 \le j \le m$. 
Since we can use as a subroutine that standard sumset computation is in near-linear output-sensitive time, it is natural to compute sumsets $A_I + B_J$ for some $I \subseteq \{1,\ldots,n\}$ and $J \subseteq \{1,\ldots,m\}$. 
Combining these two arguments, we arrive at our notion of \emph{covering}. In particular, we want to compute a covering~$\mathcal{C}$ of $(A,B,[u])$, and then use Observation~\ref{obs:sumset_from_covering} to compute $\big( \bigcup_{(I,J) \in \mathcal{C}} A_I+B_J \big) \cap [u] = (A+B) \cap [u]$ in time proportional to the cost of $\mathcal{C}$. Hence, covering algorithms are a natural class of algorithms for prefix-restricted sumset computation.

We are particularly interested in \emph{unique} coverings, because they allow us to compute prefix-restricted convolutions: A pair $A_i+B_j$ being covered exactly once corresponds to adding the product $f_i \cdot g_j$ exactly once to the output value $(f \conv g)_{i+j}$.

Moreover, we are interested in \emph{rectangle} coverings, i.e., the case of $I,J$ being intervals, because (i) they have a more geometric flavour and are thus easier to argue about, (ii) all constructions of non-rectangular coverings that we came up with could be adapted to become rectangular, so we are not aware of any better non-rectangular coverings, and (iii) for rectangular coverings we can show (unconditional) lower bounds.

\paragraph{Designing covering algorithms.}
In what follows, for ease of exposition we shall assume that $|A| = |B| = n$. Our first algorithmic step is to reduce to a promise version of the problem, where we know $\out = |(A+B) \cap [u]|$ up to a constant factor. (In fact, we could even assume to know a superset of $(A+B) \cap [u]$ whose size is at most a constant factor larger, but we do not know how to exploit such a set in this context.)

The easiest possible covering algorithm partitions $\{1,\ldots,n\}$ into $k$ consecutive intervals $I_1,\ldots,I_k$ for~$A$, and similarly into $J_1,\ldots,J_k$ for $B$. We may ignore pairs $(x,y)$ with $\min(A_{I_x}) + \min(B_{J_y}) > u$, since they do not contain any sum $A_i+B_j \le u$. For all remaining $(x,y)$, we construct the pairs $(I_x,J_y)$ to form a covering $\mathcal{C}$. In the following we discuss two ways of choosing $I_1,\ldots,I_k$ and $J_1,\ldots,J_k$.

Choosing the interval partitioning such that it splits the universe $[u]$ into parts of length $u/k$ ensures that any sum $s \in A_I + B_J$ for $(I,J) \in \mathcal{C}$ satisfies $s \le (1+\frac 2k)u$. Moreover, we show that any $s \in \mathbb{N}$ appears in at most $k$ sumsets $A_I+B_J, \, (I,J) \in \mathcal{C}$ (this holds because every ``diagonal'' of subproblems contains distinct sums). In total, we can bound the cost of $\mathcal{C}$ by $k \cdot |(A+B) \cap [(1+\frac 2k) u]|$, which for $\zeta = 2/k$ yields one result of Theorem~\ref{thm:add_relaxed_upper_bound}, see Section~\ref{sec:prefix_zeta}.

Alternatively, we can choose the interval partitioning such that it splits $A$ and $B$ into parts of size $n/k$. There are $\Oh(k)$ pairs $(I,J) \in \mathcal{C}$ along the ``boundary'', where $\min(A_I+B_J) \le u < \max(A_I+B_J)$. For each such pair we bound the cost trivially by $(\frac nk)^2$. The remaining pairs only cover sums in $(A+B) \cap [u]$, and each such sum is covered at most $k$ times, as in the previous paragraph. In total, this yields cost at most $k \cdot \out + \frac{n^2}k$, which is $\Oh(n \cdot \out^{1/2})$ by optimizing over $k$. Since $n \le \out$, in particular the cost is at most $\Oh(\out^{3/2})$. Generalizing this idea to the interval-restricted case yields Theorem~\ref{thm:add_interval_algo}, see Section~\ref{sec:interval_restricted}. For the interval-restricted case, this cost bound turns out to be (conditionally) optimal, see Theorem~\ref{thm:interval_hardness}.

\smallskip
So far, we used simple charging arguments to obtain a covering of cost $\Oh(\out^{3/2})$.
We show that for prefix-restricted sumset computation we can bypass the hardness results for the interval-restricted case and reduce the exponent further to $4/3$. 
This requires insights from Additive Combinatorics. 
Specifically, to obtain an $\tOh(\out^{4/3})$ algorithm, we make use of Ruzsa's triangle inequality. This inequality from Additive Combinatorics is a charging argument that implies
\[ |A_I+B_J| \le \frac{|A_I + B_{J'}| |A_{I'} + B_{J'}| |A_{J'} + B_J|}{|I'| |J'|}, \]
for any sets $I,I',J,J' \subseteq \{1,\ldots,n\}$. In particular, if we choose $I'$ and $J'$ such that the three sumsets on the right hand side are all subsets of $[u]$, then the size of every sumset on the right hand side is bounded from above by $\out = |(A+B) \cap [u]|$. With some care, this inequality can be used to show
\[ |A_I + B_J| \le \frac{\out^3}{\min(I) \min(J)}. \]
This allows us to bound the number of ``bad'' sums in $A_I+B_J$, namely the sums in $(A_I+B_J) \setminus [u]$, in terms of the output-size and the position of the rectangle $(I,J)$.
However, it turns out that this bound is not good enough, essentially because it yields a worst-case bound that holds for all rectangles $(I,J)$. We thus replace it with an improved bound that holds for ``most'' rectangles $(I,J)$.

Combining this bound with several tricks from our previous covering constructions 
allows us to construct a covering of cost $\tOh(\out^{4/3})$, see Section~\ref{sec:technical_core}. 

\paragraph{Connection to the Balog-Szemer\'{e}di-Gowers Theorem.} Probably the most important result on restricted sumsets is the Balog-Szemer\'{e}di-Gowers (BSG) theorem~\cite[Theorem 2.13]{tao2006additive}.  Given sets $A,B \subseteq \mathbb{Z}$ and a set of pairs $G \subseteq A \times B$, the BSG theorem is used to pass from information about the \emph{restricted} sumset $A+_{G}B = \{ a + b \mid (a,b) \in G\}$ to information about an \emph{unrestricted} sumset $A_I+B_J$, for some $I,J \subseteq \{1,\ldots,n\}$ (where we continue to assume that $|A| = |B| = n$). 
Specifically, the BSG theorem states that if $|G| = \delta n^2$ and $|A +_G B| = c n$, then there exist sets $I,J \subseteq \{1,\ldots,n\}$ of size $\Omega(\delta^2 n)$ with $|A_I+B_J| \le \Oh((c/\delta)^5 n)$. 

In a breakthrough paper, Chan and Lewenstein~\cite{CL15} algorithmically exploited the BSG theorem in order to solve several problems with additive structure. On a high level, their approach uses the BSG theorem repeatedly, constructing a sequence $(I^{(1)},J^{(1)}), \ldots, (I^{(k)},J^{(k)})$ of subsets $I^{(i)},J^{(i)} \subseteq \{1,\ldots,n\}$ and a remainder $R \subseteq A \times B$ such that
	\begin{align}
	A+_G B \;=\; \left( \huge \bigcup_{i=1}^{k} \left( A_{I^{(i)}} + B_{J^{(i)}} \right) \right) \bigcup \{ a+b: (a,b) \in R\} 	\\	
	\end{align}
	For $|A +_G B| = c n$, this decomposition satisfies that all sumsets are small, i.e., $|A_{I^{(i)}} + B_{J^{(i)}}| = \Oh((kc)^5 n)$, and the remainder is not too large, i.e., $|R| = \Oh(n^2 / k)$.
Using this decomposition, one can compute (a superset of) $A+_G B$ by computing the unrestricted sumsets $A_{I^{(i)}} + B_{J^{(i)}}$ and then iterating over all elements of $R$. Using near-linear output-sensitive sumset computation, this takes total time $\tOh(k^6 c^5 n + n^2/k)$. For several applications considered in~\cite{CL15}, this approach yields non-trivial improvements. We note, however, that constructing the sequence and the remainder  in their setting requires quadratic time in $n$.

\smallskip
A careful reader could notice that our notion of a \emph{covering} is very similar to the decomposition computed by Chan and Lewenstein~\cite{CL15}. 
In our situation, we are interested in the set $G = \{ (a,b) \in A \times B \mid a+b \le u\}$, since our goal is to compute $A +_G B = (A+B) \cap [u]$. A covering~$\mathcal{C}$ satisfies $A +_G B \subseteq \bigcup_{(I,J) \in \mathcal{C}} A_I + B_J =: T$, from which we can compute $A +_G B$ as $T \cap [u]$. Except for the remainder set $R$, this is analogous to Chan and Lewenstein's decomposition. 
Thus, both notions produce a certain form of covering of $A+_G B$.\footnote{We remark that Chan and Lewenstein~\cite{CL15} construct their decomposition for a somewhat different reason, since they are not (immediately) concerned with output-sensitivity.} 
Unfortunately, in our situation using their decomposition as a black box is worse than the simple covering of cost $\Oh(\out^{3/2})$ for two reasons. The first is that the time needed to construct their decomposition is $\Theta(|A||B|)$, which could be up to $\out^2$. The second reason is that the bound they get on the cost of the covering they produce is strictly worse than $\out^{3/2}$. Thus, even if their decomposition was given to us for free, it would be worse than the easy approach which uses elementary charging arguments.

Our approach in this paper can be seen as using the arguments of the proof of the BSG theorem in an ad-hoc manner for our specific set $G$. Indeed, at the heart of the BSG theorem lies a charging argument that is analogous to Ruzsa's triangle inequality (see~\cite[Section 7.3]{CL15}), and careful inspection of the algorithmic version of the BSG theorem (see~\cite[Sections 2 and 7]{CL15}) shows that it can be modified to produce a rectangle covering for our set $G$ with no loss in parameters.
We present a highly optimized variant of this ad-hoc construction in this paper.

In summary, our algorithms can be viewed as ad-hoc BSG-type theorems for a special choice of the restriction set $G$, obtaining bounds that seem unreachable in the more general setting.

\paragraph{Connection to Freiman-type Theorems.} Another celebrated result among combinatorialists is Freiman's theorem~\cite[Theorem 5.33]{tao2006additive}: If $A\subseteq \mathbb{N}$ satisfies $|A+A| \leq c |A|$, then $A$ is contained in an $f(c)$-dimensional arithmetic progression of length $g(c) \cdot |A|$. 
Such a result could potentially be useful in our situation, since for a rectangle with $A_I+B_J \subseteq [u]$ we have $|A_I+B_J| \le \out$, and thus a small output-size could imply that $A_I$ and $B_J$ are essentially arithmetic progressions, which is easy to exploit.
However, the state-of-the-art bounds on $f(\cdot), g(\cdot)$ are exponential and thus not useful in applications. This is why much of Additive Combinatorics has focused on bypassing the need for Freiman's theorem, to prove robust theorems with polynomial parameter dependence. 

Recently, \cite{shao2019robust} announced such a robust extension of another theorem by Freiman, called the $3k-4$ theorem, which roughly states that if $|A+_G B|$ is ``small'' and $|G|$ is ``large'' ($|G| \geq (1-\epsilon) |A||B|$), then one can extract highly non-trivial information about $A$ and $B$. However, the ``small'' and ``large'' conditions are too restrictive for our algorithmic applications, and it is unclear to us at the moment to what extent this result can be exploited algorithmically. 
We expect that further interaction with the field of Additive Combinatorics yields more insights that lead to new algorithmic machinery.

\paragraph{Selection from $\pmb{X+Y}$.} A problem that seems related at first glance is selecting the $k$-th element from $X+Y$ and row-sorted matrices, see~\cite{KKZZ19} and references therein. The crucial difference is that in that line of research $X+Y$ is treated as a~\emph{multiset}, instead of a set. For example, they consider $\{1,2\} + \{1,2,3\} = \{2,3,3,4,4,5\}$, while we define $\{1,2\} + \{1,2,3\} = \{2,3,4,5\}$. Note that in our situation the presence of multiplicities is exactly what we want to avoid. It is vital for our algorithm to process every element in the sumset the minimum number of times possible, and not proportional to the number of times it appears, since we measure running time with respect to the size of the sumset (without multiplicities). For this reason, the ideas of~\cite{KKZZ19} do not seem applicable for top-k and prefix-restricted convolution.

\paragraph{Lower Bound on Rectangle Coverings.} 
Let us also briefly discuss the construction of instances $(A,B,[u])$ on which any rectangle covering has superlinear cost (see Theorem~\ref{thm:covering_lb}).
We build such instances in two steps.

First, we construct ``many'' sets $X^{(1)},\ldots, X^{(g)}$ and $Y^{(1)},\ldots,Y^{(g)}$ such that $|X^{(\ell)} + Y^{(\ell)}|$ is ``large'', while $|X^{(\ell)} + Y^{(\ell')}|$ is ``small'' for any $\ell \neq \ell'$. To this end, we pick appropriate parameters $m,t$, and greedily choose an error-correcting code over $\{0,1\}^t$ where every codeword has Hamming weight exactly $t/2$. We define $X^{(\ell)}$ to be the set of all $t$-digit numbers in the $m$-ary number system with the constraint that their $r$-th digit is $0$ if the $\ell$-th codeword has a $1$ at position $r$. Similarly, we define $Y^{(\ell)}$ to be the set of all $t$-digit numbers in the $m$-ary system with the constraint that their $r$-th digit is $0$ if the $\ell$-th codeword has a $0$ at position $r$.

In the second step, we contruct the set $A$ as a union over several shifted copies of each $X^{(\ell)}$. The set $B$ is constructed similarly from the sets $Y^{(\ell)}$. 
In this construction, we ensure that $(A+B)\cap [u]$ contains no sumset $X^{(\ell)} + Y^{(\ell)}$, so the output-size is ``small''. We also ensure that any rectangle~$(I,J)$, that covers many parts of $(A+B)\cap [u]$ at once, needs to cover a sumset $X^{(\ell)} + Y^{(\ell)}$, which results in a ``large'' sumset $|A_I+B_J|$. Therefore, any covering either contains a rectangle of ``large'' cost or consists of ``many'' rectangles; in both cases we obtain a superlinear lower bound on the cost in terms of the output-size $|(A+B)\cap [u]|$.

\subsection{Organization}

In Section~\ref{sec:top_k_conv_reductions} we reduce top-$k$-convolution and related problems to finding a covering, proving Corollary~\ref{cor:top_k_and_related_algos}. 
In Section~\ref{sec:learn_out} we show how to learn the output-size up to a constant factor for restricted convolution problems. In Section \ref{sec:interval_restricted} we study interval-restricted convolution, proving Theorems~\ref{thm:add_interval_algo} and~\ref{thm:interval_hardness}.
In Section~\ref{sec:prefix_zeta}, we relax the upper bound, proving part of Theorem~\ref{thm:add_relaxed_upper_bound}.
In Section~\ref{sec:technical_core} we describe and analyze our $\tilde{O}(\out^{4/3})$-cost covering for prefix-restricted sumset computation (Theorem~\ref{thm:technical_core}) and we prove a lower bound for rectangle coverings (Theorem~\ref{thm:covering_lb}). Section~\ref{sec:subset_sum} then contains the reduction from \textsc{SubsetSum} to prefix-restricted sumset computation, proving Theorem~\ref{thm:subsetsum_result} and finishing the proof of Theorem~\ref{thm:add_relaxed_upper_bound}.

%% file: topkconvolution.tex
\section{Top-$\pmb{k}$-Convolution}
\label{sec:top_k_conv_reductions}

In this section, we present some easy reductions among Top-$k$-Convolution and related problems, proving Corollary~\ref{cor:top_k_and_related_algos}.
We start by proving that Top-$k$-Convolution is equivalent to Prefix-Restricted convolution.

\begin{lemma} \label{lem:equ_topk_prefix}
  If Top-$k$-Convolution on non-negative vectors can be solved in time $\tOh(k^\alpha)$, then Prefix-Restricted convolution on non-negative vectors can be solved in time $\tOh(\out^\alpha)$, and vice versa.
\end{lemma}
\begin{proof}
  To solve Prefix-Restricted sparse convolution, we run Top-$k$-Convolution for $k=2^0,2^1,2^2,\ldots$ until we reach a value of $k$ so that the $k$-th non-zero entry lies above $u$. Then we have found all elements of the Prefix-Restricted convolution, and we only computed at most twice as many non-zero entries as necessary.

  To solve Top-$k$-Convolution on vectors $f,g \in \mathbb{R}_{\ge 0}^d$, we perform binary search over $\{0,1,\ldots,2d\}$, to find the smallest $u$ such
that Prefix-Restricted convolution on $[u]$ returns a $k$-sparse vector. To obtain the desired running time we add a small twist: if the execution time of Prefix-Restricted sumset computation on $[u]$ is more than $k^\alpha \cdot \polylog d$, i.e., more than what it would be for output size $k$, then we abort and look for a smaller $u$. 

  Both reductions only add a log-factor to the running time.
\end{proof}

We next show that coverings not only allow us to compute Prefix-Restricted sumsets, but they even enable the computation of Prefix-Restricted convolutions, if the covering is \emph{unique}.

\begin{lemma} \label{lem:conv_from_covering}
  Suppose that, given $(A,B,u)$, we can compute a \emph{unique covering} of $(A,B,[u])$ with cost $\tOh(\out^\alpha)$ in expected time $\tOh(\out^\alpha)$, where $\out := |(A+B)\cap [u]|$. Then Prefix-Restricted convolution on non-negative vectors can be solved in expected time $\tOh(\out^\alpha)$.
\end{lemma}
\begin{proof}
  For a vector $f$ we denote by $f_S$ the same vector where every entry outside of $S$ is zeroed out.
  Given non-negative vectors $f,g$, let $A := \supp(f)$ and $B := \supp(g)$. The output-size of Prefix-Restricted convolution on $f,g$ is the number of non-zero entries of $(f \conv g)_{[u]}$. This is the same as the output-size of Prefix-Restricted sumset computation on $A,B$, namely $|(A+B) \cap [u]|$. Thus, the two output-sizes coincide and we denote both by ``$\out$''. In particular, we can afford to compute a unique covering $\mathcal{C}$ of $(A,B,[u])$. 
  
  Using output-sensitive convolution of non-negative vectors (Theorem~\ref{thm:cole_hariharan},~\cite{CH02}), we can compute $f_S \conv g_T$ in time $\tOh(\| f_S \conv g_T \|_0 )$. Therefore, we can compute 
    \[ h := \sum_{(I,J) \in \mathcal{C}} f_{A_I} \conv g_{B_J}, \]
    in time proportional to the cost of the covering, up to log factors. By the properties  of a unique covering, any non-zero product $f_j \cdot g_{i-j}$ appears exactly once in the definition of $h$, for any $0 \le j \le i \le u$.
    We thus have $h_{[u]} = (f \conv g)_{[u]}$.
\end{proof}

Corollary~\ref{cor:top_k_and_related_algos} now follows.

\begin{proof}[Proof of Corollary~\ref{cor:top_k_and_related_algos}]
  To solve Prefix-Restricted sumset computation in time $\tOh(\out^{4/3})$, we combine our covering construction (Theorem~\ref{thm:technical_core}) with Observation~\ref{obs:sumset_from_covering}. To solve Prefix-Restricted convolution on non-negative vectors in time $\tOh(\out^{4/3})$, we combine our covering construction (Theorem~\ref{thm:technical_core}) with Lemma~\ref{lem:conv_from_covering}. To solve Top-$k$-Convolution on non-negative vectors in time $\tOh(k^{4/3})$, we combine the result for Prefix-Restricted convolution with the equivalence shown in Lemma~\ref{lem:equ_topk_prefix}.
\end{proof}

%% file: sparse_convolutions.tex
\input{conv_reduction}

\input{interval_restricted}

\input{prefix_restricted}

\input{covering_lb}

%% file: conv_reduction.tex

\section{Restricted Sumset Computation: Learning the Output-Size} \label{sec:learn_out}

In this section, we show that we can assume to know the output-size up to a constant factor. 
We will later use this for Prefix-Restricted and for Interval-Restricted sumset computation. Here we will discuss the more general setting of Interval-Restricted sumset computation; the same construction works for the Prefix-Restricted case.

Suppose that we are given sets $A,B$ and integers $\ell,u$ and we want to compute their Interval-Restricted sumset $(A+B) \cap [\ell,u]$. We show that we can assume to know $|(A+B)\cap [\ell,u]|$ up to a constant factor.
In fact, we reduce Interval-Restricted sumset computation to a seemingly much easier variant, where additionally we are given a set $T$ of size $ \Theta(|(A+B)\cap[\ell,u]|)$ which contains $(A+B)\cap[\ell,u]$. 

\begin{definition}[Interval-Restricted Sumset Computation with a Promise (IR-SMP)]
Given sets $A, B \subseteq \mathbb{Z}$, numbers $\ell \le u$, and a set $T$ of size $|T| \leq 6|(A+B)\cap[\ell,u]| + 12$, such that 
\[ (A+B) \cap [\ell,u] \subseteq T, \]
 compute the set
\[	(A+B) \cap [\ell,u].	\]
\end{definition}

We present a reduction from the problem variant without promise to the variant with promise.

\begin{lemma} \label{lem:conv_promise}
  We can reduce a given instance $(A,B,\ell,u)$ of Interval-Restricted sumset computation to $\Oh(\log(u-\ell))$ instances of IR-SMP, each of the form $(A',B',\ell',u',T')$ with $|A'| \le |A|$, $|B'| \le |B|$, $\ell' \le \ell$, $u' \le u$, $u'-\ell' \le u-\ell$, and $|T'| = \Oh(|(A+B)\cap[\ell,u]|)$. The reduction runs in linear time in its output-size.
\end{lemma}

\begin{proof}
In what follows, for integers $x,y$ we define $x \div y = \lfloor \frac{x}{y} \rfloor$. For a set $Z$ we write $Z \div x = \{ z \div x \mid z\in Z\}$, and we write $x Z = \{x \cdot z \mid z \in Z\}$.

Let $(A,B,\ell,u)$ be a given instance of Interval-Restricted sumset computation.
We compute better and better approximations of the desired output $(A+B) \cap [\ell,u]$ by computing the sets 
\[ S^{(i)} \;:=\; \big((A \div 2^i) + (B \div 2^i)\big) \cap \big[\ell \div 2^i, u \div 2^i\big], \]
for $i = r, \ldots,0$, where $r := \lceil \log(u-\ell) \rceil$.
See Algorithm~\ref{alg:learn_out} for pseudocode.
Note that $S^{(0)}$ is the desired output $(A+B) \cap [\ell,u]$. 
We compute each set $S^{(i)}$ by calling an IR-SMP oracle on a suitable promise~$T^{(i)}$, determined as follows. 

\begin{algorithm}[!hb]
\caption{Reduction of Interval-Restricted Sumset Computation to its Promise Version}
\label{alg:learn_out}
\begin{algorithmic}[1]
\Procedure{\textsc{ComputeIntervalRestrictedSumset}}{$A,B,\ell,u$} \\
\begin{flushright}
\Comment{$A,B \subseteq [u],\, \ell \le u$}\\
\Comment{We assume oracle access to Interval-Restricted sumset computation with a promise (IR-SMP)}
\end{flushright}
\State $r \leftarrow \left \lceil \log (u-\ell) \right \rceil $		
\State $T^{(r)} \leftarrow \left \{ \ell \div 2^r, \ldots, u \div 2^r \right \}$		
\State $S^{(r)} \leftarrow$ \textsc{IR-SMP}($A \div 2^r, B \div 2^r, \ell \div 2^r, u \div 2^r, T^{(r)}$) \label{line:promise_call} \\\Comment{ Call to the promise problem}
\For{ $i= r-1$ down to $0$ } 			
	\State $T^{(i)} \leftarrow  \big( 2 S^{(i+1)}  + \{0,1,2\} \big)$
	\State $T^{(i)} \leftarrow T^{(i)} \cup \{(\ell \div 2^i), (\ell \div 2^i)+1, (u \div 2^i)\}$  \label{line:expansion} 
	\State $S^{(i)} \leftarrow$ \textsc{IR-SMP}($A \div 2^i, B \div 2^i, \ell \div 2^i, u \div 2^i, T^{(i)}$) \label{line:promise_call} \\\Comment{ Call to the promise problem}
\EndFor
	\State Return $S^{(0)}$
\EndProcedure
\end{algorithmic}
\end{algorithm}

For $i = r$, the interval $[\ell \div 2^i, u \div 2^i]$ has length at most 2. Thus, the set $T^{(r)} := \{\ell \div 2^i, \ldots, u \div 2^i\}$ is a valid promise from which our assumed IR-SMP oracle can compute the set $S^{(r)}$.

For $i < r$, we claim that a valid promise for the instance $(A \div 2^i, B \div 2^i, \ell \div 2^i, u \div 2^i)$ is given~by
\[ T^{(i)} := \big(2 S^{(i+1)} + \{0,1,2\}\big) \cup \{(\ell \div 2^i), (\ell \div 2^i) + 1, (u \div 2^i)\}. \]
If this claim holds, then from $S^{(r)}$ we can compute $S^{(r-1)}, S^{(r-2)},\ldots$, $S^{(0)}$, finishing our reduction.

Let us also claim that any set $S^{(i)}$ has size $|S^{(i)}| = \Oh(|(A+B)\cap[\ell,u]|)$.
Then in particular the set $T^{(i)}$ has size $|T^{(i)}| = \Oh(|S^{(i)}|) = \Oh(|(A+B)\cap[\ell,u]|)$. Hence, all constructed instances satisfy the claimed properties, and we reduced the given instance $(A,B,\ell,u)$ to $\Oh(\log(u-\ell))$ calls to an IS-SMP oracle. This finishes the proof, except that it remains to prove the two claims.

\begin{claim} \label{cla:sizeS}
  The set $S^{(i)}$ has size $|S^{(i)}| \le 2 |(A+B)\cap[\ell,u]| + 3$.
\end{claim}
\begin{proof}
  It is easy to check that for any integers $a,b$ we have
  \begin{align}
    2^i\big( (a \div 2^i) + (b \div 2^i) \big) \;\le\; a+b \;<\; 2^i\big( (a \div 2^i) + (b \div 2^i) \big) + 2\cdot 2^i,  \label{eq:sfhaskf}
  \end{align}
  or, equivalently,
  \begin{align}
    \frac{a+b}{2^i}-2 \;<\; (a \div 2^i) + (b \div 2^i) \;\le\; \frac{a+b}{2^i}. \label{eq:osvavsv}
  \end{align}
  Any number in $s \in S^{(i)}$ can be expressed as $s = (a \div 2^i) + (b \div 2^i)$ with
  \[ (\ell \div 2^i) \le s \le (u \div 2^i). \]
  Suppose that we have $s \in [(\ell \div 2^i)+1, (u \div 2^i)-2]$.
  Then by inequalities (\ref{eq:sfhaskf}), it follows that
  \[ \ell \;\le\; 2^i((\ell \div 2^i)+1) \;\le\; 2^is \;\le\; a + b \;<\; 2^is+2\cdot 2^i \;\le\; 2^i(u \div 2^i) \;\le\; u, \]
  Thus, we obtain that the number $s' := a+b$ lies in $(A+B) \cap [\ell,u]$. 
  We say that \emph{$s$ charges~$s'$}.
  From inequalities~(\ref{eq:osvavsv}), we infer that any number $s' \in (A+B) \cap [\ell,u]$ can only be charged by numbers in $(s'/2^i-2,s'/2^i]$, so it is charged by at most two numbers $s \in S^{(i)}$. 
  This charging scheme proves
  \[ \big|S^{(i)} \cap [(\ell \div 2^i)+1, (u \div 2^i)-2]\big| \;\le\; 2 \cdot \big|(A+B) \cap [\ell,u]\big|. \]
  We add 3 to account for the numbers $\ell \div 2^i$, $u \div 2^i - 1$, and $u \div 2^i$. This yields the claim.
\end{proof}

\begin{claim}
  The set $T^{(i)}$ is a valid promise for the instance $(A \div 2^i, B \div 2^i, \ell \div 2^i, u \div 2^i)$. That is, we have $T^{(i)} \supseteq S^{(i)}$ and $|T^{(i)}| \le 6 |S^{(i)}| + 12$.
\end{claim}
\begin{proof}
  Since $x \div 2^{i+1} = (x \div 2^i) \div 2$, it suffices to prove the claim for $i=0$. The same proof then also works for larger $i$ after replacing $A$ by $A \div 2^i$, $B$ by $B \div 2^i$, $\ell$ by $\ell \div 2^i$, and $u$ by $u \div 2^i$. 
  
  We first show that $|T^{(0)}| \le 6 |S^{(0)}| + 12$. 
  Recall that in Claim~\ref{cla:sizeS} we proved that
  \[ |S^{(1)}| \;\le\; 2\cdot |(A+B) \cap [\ell,u]| + 3 \;=\; 2 |S^{(0)}| + 3. \] 
  We combine this with the inequality $|T^{(0)}| \le 3 |S^{(1)}| + 3$ that we obtain from the construction of~$T^{(0)}$. This yields the claimed inequality $|T^{(0)}| \le 6 |S^{(0)}| + 12$. 
  
  Next we prove $T^{(0)} \supseteq S^{(0)}$.
  Consider any $a \in A,\, b \in B$ with $a+b \in \big[\ell+2,u-1\big]$.
  Using inequalities~(\ref{eq:osvavsv}) we obtain
  \begin{align*}
    \ell \div 2 \;\le\; \frac \ell{2} \;\le\; \frac{a+b}{2} - 1 &\;\le\; (a \div 2)+(b \div 2)  \\
    &\;\le\; \frac{a+b}{2} \;\le\; \frac{u-1}{2} \;\le\; u \div 2. 
   \end{align*}
   Therefore, the sum $(a \div 2)+(b \div 2)$ is in $S^{(1)}$. Now, since $a+b$ can be found among the three numbers
   \[ 2\big((a \div 2) + (b \div 2)\big), \;\; 2\big((a \div 2) + (b \div 2)\big) + 1, \;\; 2\big((a \div 2) + (b \div 2)\big) + 2, \]
  it follows that $a+b \in (2S^{(1)} + \{0,1,2\}) \subseteq T^{(0)}$. 
  Recall that here we assumed $a+b \in [\ell+2,u-1]$. The boundary numbers $\ell, \ell+1$, and $u$ are handled by explicitly adding them to the set $T^{(0)}$ in its construction. We thus obtain $S^{(0)} \subseteq T^{(0)}$.
\end{proof}
These claims finish the proof of Lemma~\ref{lem:conv_promise}.
\end{proof}

We obtain the following easy corollary where we essentially ignore the superset $T$ and only keep its size $|T|$, which is a constant-factor approximation of the output-size. 

\begin{lemma}[Interval-Restricted Sumset Computation with Approximate Output-size] \label{lem:knowout}
  Suppose that given $(A,B,\ell,u)$ and an additional input $\tilde{\out}$ satisfying $\out \le \tilde{\out} \le 6 \out + 12$, we can compute a covering of $(A,B,[\ell,u])$ of cost $O(c)$ in time~$O(T)$, where $c$ and $T$ are monotone functions of $|A|,|B|,u,\out$. Then Interval-Restricted sumset computation can be solved in time $\tOh(c+T)$.
  
  An analogous statement holds for Prefix-Restricted sumset computation.
\end{lemma}
\begin{proof}
  Given an instance $(A,B,\ell,u)$ of Interval-Restricted sumset computation, we run Lemma~\ref{lem:conv_promise} to reduce to $\Oh(\log(u-\ell))$ promise instances of the form $(A',B',\ell',u',T')$. By the properties of $T'$, for $\tilde{\out}' := |T'|$ we have $\out' \le \tilde{\out'} \le 6\out' +12$, where $\out' := |(A'+B') \cap [\ell',u']|$. Hence, we can use our assumed algorithm to compute a covering of $(A',B',[\ell',u'])$ of cost $O(c)$ in time~$O(T)$. By Observation~\ref{obs:sumset_from_covering}, we can thus solve the instance $(A',B',[\ell',u'])$ in time $\tOh(c+T)$. Over $\Oh(\log(u-\ell))$ many constructed instances, we obtain the same time bound up to log-factors.
\end{proof}

%% file: interval_restricted.tex

\section{Interval-Restricted Sumset Computation} \label{sec:interval_restricted}

In this subsection we prove Theorems~\ref{thm:add_interval_algo} and \ref{thm:interval_hardness}. 

\subsection[An $\tilde{O}(\sqrt{mn \cdot \out})$-Time Algorithm]{An \boldmath$\tilde{O}(\sqrt{mn \cdot \out})$-time Algorithm}

In this subsection we prove Theorem~\ref{thm:add_interval_algo}. 

\begin{proof}
Given $A,B,\ell,u$, we will compute the Interval-Restricted sumset $(A+B) \cap [\ell,u]$ in time $\tilde{O}(n+m+\sqrt{mn \cdot \out})$. More precisely, we will compute a covering of $(A,B,[\ell,u])$ of cost $\Oh(\sqrt{mn \cdot \out})$ in time $\Oh(n+m)$, assuming that we know an approximation $\tilde{\out}$ of the output-size satisfying $\out \le \tilde{\out} \le O(\out)$ (we can assume this by Lemma~\ref{lem:knowout}). 

Set $q := \big\lceil \big(nm/\tilde{\out}\big)^{1/2} \big\rceil$. Note that since $n,m \le \out \le n\cdot m$ we have $1 \le q = \Oh(\min\{n,m\})$. 
We assume that $n$ and $m$ are divisible by $q$, which we can ensure by duplicating at most $q$ elements in $A$ and $B$ (with the slight abuse of transitioning to multi-sets). 

\medskip
We split $A$ and $B$ into $q$ subsets by defining for any $1 \le i,j \le q$:
\begin{align*}
  I_i &:= \{ (i-1) \cdot n/q + 1,\ldots, i \cdot n/q\}, \\
  J_j &:= \{ (j-1) \cdot m/q + 1,\ldots, j \cdot m/q\}, \\
  A^{(i)} &:= A_{I_i}, \\
  B^{(j)} &:= B_{J_j}.
\end{align*}
We let $\mathcal{C}$ be the set of all pairs $(I_i,J_j)$ with 
\[ \min(A^{(i)}) + \min(B^{(j)}) \le u \quad\text{and}\quad \max(A^{(i)}) + \max(B^{(j)}) \ge \ell. \]
Observe that $\mathcal{C}$ is a unique-rectangle-covering of $(A,B,[\ell,u])$. 
We can easily compute $\mathcal{C}$ by brute force, by iterating over all $1 \le i,j \le q$ and testing whether to put $(I_i,J_j)$ into $\mathcal{C}$ in time $\Oh(1)$, see Algorithm~\ref{alg:find_covering}. This takes total time $\Oh(q^2) = \Oh(nm/\out) = \Oh(\min\{n,m\})$, assuming that we have random access to $A$ and $B$.

\begin{algorithm}[!hb] 
\caption{Covering for Interval-Restricted Sumset Computation}
\label{alg:find_covering}
\begin{algorithmic}[1]
\Procedure{\textsc{FindCovering}}{$A,B,\ell,u,\tilde{\out}$} 
\State $q \leftarrow \big\lceil \big(nm/\tilde{\out}\big)^{1/2} \big\rceil$
\State $\mathcal{C} \leftarrow  \emptyset$
\For { $i=1$ to $q$}
	\For {$j=1$ to $q$}
	\State $I \leftarrow \{(i-1) \cdot n/q +1 ,\ldots,i \cdot n/q\}$
	\State $J \leftarrow \{(j-1) \cdot m/q + 1,\ldots, j \cdot m/q\}$
	\If {$\min(A_{I_i}) + \min(B_{J_j}) \le u$}
	 \If{$\max(A_{I_i}) + \max(B_{J_j}) \ge \ell$} 
		\State $\mathcal{C} \leftarrow \mathcal{C} \cup \{(I,J)\}$
	\EndIf
	\EndIf
	\EndFor
\EndFor

\State Return $\mathcal{C}$
\EndProcedure
\end{algorithmic}
\end{algorithm}

Recall that the cost of a covering $\mathcal{C}$ is
\[ \sum_{(I,J) \in \mathcal{C}} |A_{I}+B_{J}|. \]
We split $\mathcal{C}$ into two parts, the rectangles in the interior and the rectangles at the boundary:
\begin{align*}
  \mathcal{C}_{\textup{int}} &:= \{(I,J) \in \mathcal{C} \mid A_I+B_J \subseteq [\ell,u] \}, \\
  \mathcal{C}_{\textup{bd}} &:= \mathcal{C} \setminus \mathcal{C}_{\textup{int}}.
\end{align*}

For the interior rectangles, we split their cost into diagonal sums of the form
\[ \sum_{(I_i,J_{i+\Delta}) \in \mathcal{C}_{\textup{int}}} |A^{(i)}+B^{(i+\Delta)}|, \]
for $-q < \Delta < q$.
We claim that each such diagonal sum is bounded from above by $\out$. Indeed, for consecutive terms along a diagonal we have
\[ \max(A^{(i)}) + \max(B^{(i+\Delta)}) < \min(A^{(i+1)}) + \min(B^{(i+\Delta+1)}). \]
Therefore, the output-sizes $|A^{(i)}+B^{(i+\Delta)}|$ are disjoint contributions to $\out$, for fixed $\Delta$ and ranging over all $i$. It follows that a diagonal sum is bounded by $\out$, and since there are $2q-1$ diagonals, we obtain
\[ \sum_{(I,J) \in \mathcal{C}_{\textup{int}}} |A_{I}+B_{J}| \le (2q-1)\out = \Oh(\sqrt{nm \cdot \out}).\]

We argue geometrically about the boundary. Observe that $\mathcal{C}_{\textup{bd}}$ contains at most two rectangles per diagonal, which yields $|\mathcal{C}_{\textup{bd}}| \le 4q$. For each $(I,J) \in \mathcal{C}_{\textup{bd}}$ we use the trivial upper bound $|A_I+B_J| \le |A_I| \cdot |B_J| = n m/q^2$. In total, this yields cost
\[ \sum_{(I,J) \in \mathcal{C}_{\textup{bd}}} |A_{I}+B_{J}| \le 4q \cdot \frac{nm}{q^2} = \Oh(\sqrt{n m \cdot \out}). \]

The contribution from both parts is the same, so in total we bounded the cost of $\mathcal{C}$ by $\Oh(\sqrt{n m \cdot \out})$. This finishes the proof.
\end{proof}

\subsection{Hardness Results}

In this subsection we prove Theorem \ref{thm:interval_hardness}.
\begin{proof}
We want to prove hardness of Interval-Restricted sumset computation.
Note that here we analyze running time in terms of $n=|A|,\, m=|B|$, and $\out$, which are all invariant under shifting $A$ by adding a number $q$ to each $a \in A$; similarly for $B$.
Therefore, we may drop the assumption that $A,B$ are sets of positive integers and let them be subsets of $\mathbb{Z}$ instead. This is the case because we can shift $A,B$ and the interval appropriately, so that every number in the input becomes non-negative.

\paragraph{Reduction from Boolean Matrix Multiplicaion to Interval-Restricted Sumset Computation:}
In Boolean matrix multiplication we are given $n \times n$ matrices $\overline A, \overline B$ with entries in $\{0,1\}$ and want to compute their product $C$ with $C_{ij} = \bigvee_r \overline A_{ir} \wedge \overline B_{rj}$.

Given matrics $\overline A, \overline B$, we construct sets $A, B$ as
\begin{align*}
A &:= \{  r M^2 + \overline A_{ir} \cdot M + i \mid  i ,r \in [n] \}, \\
B &:= \{  -r M^2 + \overline B_{rj} \cdot M + n j \mid r,j \in [n] \},
\end{align*}
where $M$ is any integer greater than $10(n^2 + n)$. We also set 
\[ \ell := 2M + n + 1,\quad u := 2M + n^2 + n. \]

We observe the following.
\begin{enumerate}
\item Every integer of the form $( \overline A_{ir} + \overline B_{rj} )\cdot M  + (i + n j )$ with $i,r,j \in [n]$ is contained in $A+B$.
\item For $ r \neq r'$ any sum $( r M^2 + \overline A_{ir} \cdot M + i ) + ( -r' M^2 + \overline B_{r'j} \cdot M + nj )$ is either less than  $-M^2 + 2M + n^2 + n < 0$ or at least $ M^2 + n + 1$, and hence outside of $[\ell,u]$.
\item If $\overline A_{ir} \wedge \overline B_{rj}=1$, then  $( \overline A_{ir} + \overline B_{rj} ) M  + (i + n  j ) = 2M + (i+nj)$.
\item If $\overline A_{ir} \wedge \overline B_{rj}=0$, then $( \overline A_{ir} + \overline B_{rj} ) M  + (i + n j ) \le M + (i+nj) < 2M$.
\end{enumerate}

It follows that from $(A+B) \cap [2M + n+1 , 2M + n^2 + n]$ we can infer all entries of the product matrix $C$.

Note that the output-size is $\out \le u-\ell+1 = n^2$. Hence, for any $\delta >0$, an $\Oh((|A|+|B|+\out)^{\omega/2 - \delta})$-time algorithm for Interval-Restricted sumset computation would yield an $O(n^{2\omega - 2\delta})$-time algorithm for Boolean matrix multiplication.
%

\paragraph{Reduction from Sliding Window Hamming Distance to Interval-Restricted Convolution:}
Note that Interval-Restricted convolution allows us to not only to compute $(A+B)\cap[\ell,u]$, but also the number of ways an element $x \in (A+B)\cap [\ell,u]$ can be written as $x = a +b$ with $a\in A,\, b \in B$; let us call this number the multiplicity of $x$. 
Indeed, for the indicator vectors of $A$ and $B$, the $x$-th entry of their convolution is the multiplicity of $x$ in $A+B$.

In the sliding window Hamming distance problem we are given a text $t$ of length $2n$ and a pattern $p$ of length $n$ and want to compute the Hamming distance between the pattern and every length-$n$ substring of the text. Given such an instance $t,p$, we construct sets $A,B$ as
\[	A := \{ M \cdot t_i + i \mid 1 \le i \le 2n\}, \quad B := \{- M \cdot p_j - j \mid 1 \le j \le n\},	\]
where $M := 100n$.
We also set $\ell:= 1,\, u := n$ and compute, as mentioned in the previous paragraph, the Interval-Restricted sumset $(A+B) \cap [\ell,u]$ as well as the multiplicity of every $x$ in this set. 

Fix a $1 \le i \le n$ and observe that
\begin{enumerate}
\item If $t_{i+j} = p_j$ then the the pair $(i+j,j)$ will contribute $1$ to the multiplicity of $i$.
\item If $t_{i+j} \neq p_j$ then the pair $(i+j,j)$ will contribute 1 to the multiplicity of a coefficient outside of the interval $[1,n]$, by the choice of $M$.
\item Every pair $(i',j)$ with $i' \le j$ contributes 1 to a coefficient outside of the interval $[1,n]$.
\end{enumerate}
It follows that we can read off the Hamming distance between the pattern $p$ and the $i$-th length-$n$ substring of $t$ from the multiplicity of $i$ in the output. This completes the reduction from sliding window Hamming distance to prefix-restricted convolution.
Since $\out \le n$, any $\tOh(|A|+|B|+(|A|\cdot |B| \cdot\out)^{1/2-\delta})$-time algorithm for Interval-Restricted comvolution would solve sliding window Hamming distance in time $\tOh(n^{3/2-3\delta})$.
\end{proof}

%% file: prefix_restricted.tex

\section{Relaxed Version of Prefix-Restricted Convolution} \label{sec:prefix_zeta}

We show how to solve Prefix-Restricted convolution on any instance $(A,B,u)$ in time $\tilde{O}(\zeta^{-1} |(A+B)\cap[u(1+\zeta u)]| + \zeta^{-2})$, for any $\zeta\leq 1$, proving Theorem \ref{thm:add_relaxed_upper_bound}. More precisely, we prove the following theorem, from which we conclude Theorem~\ref{thm:add_relaxed_upper_bound} using Observation~\ref{obs:sumset_from_covering}.

\begin{theorem} \label{thm:zeta_cost}
Given sets $A,B$ and a target $u$ we can compute a unique-rectangle-covering of $(A,B,[u])$ with cost 
\[ O(\zeta^{-1}|(A+B)\cap [(1+\zeta) u]|)\] 
in time $O(|A| + |B| + \zeta^{-2})$.
\end{theorem}

\begin{proof}
Find the smallest $\ell$ such that $2^{-\ell} \leq \zeta$. It suffices to solve the problem for $\zeta = 2^{-\ell}$, since $(A+B) \cap [u + 2^{-\ell} u] \subseteq (A+B) \cap [u+ \zeta u ]  $, and $2^\ell \leq 2 \zeta^{-1}$. Thus, we can assume that $1/\zeta$ is a power of 2. Moreover, by shifting $u$ as well as $A$ by a suitable number we can assume that $u$ is a power of $2$.
We split the set $[u]$ into the intervals
\[ U_r := \left[ (r-1) \cdot (\zeta u)/2 +1, r \cdot (\zeta u)/2 \right]. \]
Moreover, we define the following sets for $1 \le i,j \le 2/\zeta$:
\begin{align*}
A^{(i)} &:= A \cap U_i,  \\
B^{(j)} &:= B \cap U_j,  \\
I_i &:= \{ i' \in [n] : A_{i'} \in U_i \}  \\
J_j &:= \{ j' \in [m]: B_{j'} \in U_j \}.
\end{align*}


Let $\mathcal{C}$ be the set of all pairs $(I_i,J_j)$ for which $(A^{(i)}+B^{(j)}) \cap [u]$ is not empty; this condition can be easily checked by testing whether
\[\mathrm{min}(A^{(i)}) + \mathrm{min}(B^{(j)}) \leq u.\] 
Note that $\mathcal{C}$ can be computed in time $O(\zeta^{-2})$, assuming that we have random access to $A$ and $B$.
%
%
Observe that $\mathcal{C}$ is indeed a unique-rectangle-covering of $(A+B)\cap[u]$. 
Recall that the cost of $\mathcal{C}$ is 
	\[ \sum_{(I,J) \in \mathcal{C}} |A_I + B_J|.		\]
Similarly to the argument in the proof of Theorem \ref{thm:add_interval_algo}, this sum can be decomposed into $4/\zeta$ diagonal sums of the form 
\[ \sum_{i} |A^{(i)}+B^{(i+\Delta)}|,	\]
where the sum is over all $i$ such that $(I_i,J_{i+\Delta}) \in \mathcal{C}$.
We again use 
\[ \max(A^{(i)})+\max(B^{(j)}) < \min(A^{(i+1)})+\min(B^{(j+1)}). \] 
Moreover, for any $(i,j) \in P$ we now have
\begin{align*}
  \max(A^{(i)}) + \max(B^{(j)})
  &\le \Big( \min(A^{(i)}) + \frac{\zeta u}2 \Big) + \Big( \min(B^{(j)}) + \frac{\zeta u}2 \Big) \\
  &= \Big( \min(A^{(i)}) + \min(B^{(j)}) \Big) + \zeta u \\
  &\le u + \zeta u.
\end{align*}
It follows that every diagonal sum contributes at most \[\left|(A+B) \cap [u + \zeta u] \right|.\]
Summing over all diagonals, in total we can bound the cost of $\mathcal{C}$ by $\Oh( \zeta^{-1} \left|(A+B) \cap [u + \zeta u] \right| )$.
\end{proof}

\section{Construction of the $\mathbf{\tilde{O}(\out^{4/3})}$-cost Covering} \label{sec:technical_core}

This section is devoted to proving the technical core of our Prefix-Restricted sumset algorithm, specifically we prove Theorem~\ref{thm:technical_core}.

\subsection{An Additive Combinatorics Ingredient: Ruzsa's Triangle Inequality}

The following is a classical result from Additive Combinatorics. We present a self-contained proof.


\begin{lemma}[Ruzsa's Triangle Inequality, see also {\cite[Theorem 2]{bukh}}] \label{lem:actual_ruzsa}
For any $A,B,C \subseteq \mathbb{Z}$ we have
	\[	|A-B| \leq \frac{|A-C| \cdot |C-B|}{|C|}.	\]
\end{lemma}
\begin{proof}
  We associate every $s \in A-B$ with the lexicographically smallest pair $(a,b) \in A \times B$ such that $s = a-b$, and we denote this pair by $(a(s),b(s))$. 
  Consider the mapping 
  \begin{align*}
    (A-B) \times C &\quad\to\quad (A-C) \times (C-B) \\
    (s,c) &\quad\mapsto\quad (a(s) - c, c - b(s))
  \end{align*}    
  We claim that this mapping is injective. Indeed, from an image $(x,y) = (a(s) - c, c - b(s))$ we can infer $s = a(s) - b(s) = x + y$. The value $s$ then determines $a(s)$ and $b(s)$, so we can infer $c = y + b(s)$. We thus recovered the corresponding preimage $(s,c)$.
  
  Since this mapping is injective, we obtain $|A-B| \cdot |C| \le |A-C| \cdot |C-B|$.
\end{proof}

We will use the following simple corollary.

\begin{lemma}[Corollary of Ruzsa's Triangle Inequality] \label{lem:ruzsa}
For any $X,Y,Z,W \subseteq \mathbb{Z}$ we have
	\[	|X+Y| \leq \frac{|X+Z| \cdot |Z+W| \cdot |W+Y|}{|Z| \cdot |W|}.	\]
\end{lemma}
\begin{proof}
  First use Ruzsa's triangle inequality on $A = X,\, B = -Y,\, C = -Z$ to obtain
  \[ |X+Y| \le \frac{| X+Z | \cdot | Z-Y |}{|Z|}. \]
  Then use Ruzsa's triangle inequality on $A = Z,\, B = Y,\, C = -W$ to obtain
  \[ |Z-Y| \le \frac{| Z+W | \cdot | W+Y |}{|W|}. \]
  Plugging the latter into the former proves the claim.
\end{proof}


\subsection{Description of the Algorithm}

Given $A,B \subseteq [u]$, we write $n = |A|,\, m = |B|$, and $\out = |(A+B) \cap [u]|$.
We describe an algorithm that 
computes a unique rectangle covering $\mathcal{C}$ of $(A,B,[u])$ of cost $\tOh(\out^{4/3})$ 
in time $\tOh(\out^{4/3})$. However, as a subroutine we will use output-sensitive sumset computation (Theorem~\ref{thm:sparse_convolution}), which shows that it is hard to completely separate the tasks of computing a covering and computing the Prefix-Restricted sumset itself.

Invoking Lemma \ref{lem:knowout}, it suffices to solve the promise problem where we are given a value $\tilde{\out}$ guaranteed to satisfy $\out \le \tilde{\out} \le \Oh(\out)$. 
We can assume that $\tilde{\out}$ is larger than some absolute constant, since otherwise we have $|A|,|B| \le \out \le \tilde{\out} = \Oh(1)$, so we can compute a trivial covering of cardinality 1 and cost $|A| \cdot |B| = \Oh(1)$. 


We maintain families $\mathcal{C}, \mathcal{D}$, initialized to $\mathcal{C} = \emptyset$ and $\mathcal{D} = \{([n],[m])\}$, with the invariant that $\mathcal{C} \cup \mathcal{D}$ is a unique rectangle covering of $(A,B,[u])$. 
We refer to the rectangles in $\mathcal{D}$ as the \emph{unprocessed subproblems}, or simply \emph{subproblems}. 
The algorithm is finished when there are no more unprocessed subproblems, i.e., $\mathcal{D} = \emptyset$, and then we return the unique rectangle covering $\mathcal{C}$ as output. 

We associate to every subproblem $(I,J) \in \mathcal{D}$ the \emph{type} $(x,y)$ for $x = \left\lceil \log |I| \right \rceil$ and $y = \left\lceil \log |J| \right \rceil $.
Initially, there is exactly one subproblem $([n],[m])$ of type $(\left \lceil \log n \right \rceil,\left\lceil \log m \right\rceil)$.

We define a total order on types: $(x,y) \prec (x',y')$ iff $x+y < x'+y'$, or $x+y=x'+y'$ and $x<x'$. 
Our algorithm processes types in descending order according to this total order. For any type $(x,y)$, we process all subproblems of type $(x,y)$ in one batch. 
Upon processing a subproblem, our algorithm may generate further subproblems of strictly smaller type. 
As we will see below, all subproblems of the same type that we generate are disjoint, i.e., for any subproblems $(I,J),(I',J')$ of the same type we have $I \cap I' = \emptyset$ and $J \cap J' = \emptyset$.

Since we want to process all subproblems of a particular type $(x,y)$ in one batch, we need to store the family $\mathcal{D}$ in such a way that we can efficiently enumerate all subproblems of a type $(x,y)$.
To this end, we store $\mathcal{D}$ in a standard data structure such as a self-adjusting binary search tree, where subproblems are first compared according to their type and then according to the endpoints of $I$ and $J$. Note that we can store any subproblem $(I,J)$ using $\Oh(1)$ integers, since $I$ and $J$ are intervals.

We set the parameter $q := \lceil \tilde{\out}^{1/3} \rceil$. 

To finish the description of the algorithm, it remains to describe how we process all subproblems of a particular type $(x,y)$ in one batch. We consider two cases.

\smallskip
\emph{Case 1:} $2^{x+y} \leq \tilde{\out}$. Then for each subproblem $(I,J)$ of type $(x,y)$ we move $(I,J)$ from $\mathcal{D}$ to~$\mathcal{C}$.

\smallskip
\emph{Case 2:} $2^{x+y} > \tilde{\out}$. If $\mathcal{D}$ contains more than $q$ subproblems of type $(x,y)$, then we do the following. We start computing $A_I+B_J$ for each of these subproblems in parallel (using Theorem~\ref{thm:sparse_convolution}), and we stop once all but $q$ of these calls have finished. For each finished call $A_I+B_J$, we move $(I,J)$ from $\mathcal{D}$ to $\mathcal{C}$.

At this point, we have at most $q$ subproblems of type $(x,y)$ left. We split each such subproblem $(I,J) = ([i_1,i_2],[j_1,j_2])$ as follows. We set $i := \lfloor (i_1+i_2)/2 \rfloor$ and determine the maximum index $j \in J$ with $A_i+B_j \le u$. This splits $I$ into $I_1 = [i_1,i]$ and $I_2 = [i+1,i_2]$ and $J$ into $J_1 = [j_1,j]$ and $J_2 = [j+1,j_2]$.
We add $(I_1,J_1)$ to the output $\mathcal{C}$. 
Moreover, we add the subproblems $(I_1,J_2)$ and $(I_2,J_1)$ to $\mathcal{D}$ and we remove $(I,J)$ from $\mathcal{D}$.
Since $u < A_{i}+B_{j+1} \le A_{i+1}+B_{j+1}$, we can ignore the subproblem $(I_2,J_2)$.

\smallskip
This finishes the description of our algorithm, for pseudocode see Algorithm~\ref{alg:covering_construction}.

\begin{algorithm}[!t]\caption{}\label{alg:covering_construction}
\begin{algorithmic}[1]
\Procedure{\textsc{CoveringConstruction}}{$A,B,u,\tilde{\out}$}
\State $n \leftarrow |A|$, $m \leftarrow |B|$, $q \leftarrow \tilde{\out}^{1/3}$
\State Initialize $\mathcal{C} \leftarrow \emptyset$, $\mathcal{D} \leftarrow \{([n],[m])\}$
\While {$\mathcal{D} \ne \emptyset$}
  \State Let $(x,y)$ be the largest type such that $\mathcal{D}$ contains subproblems of type $(x,y)$
  \State \emph{// process all subproblems of type $(x,y)$ in one batch:}
  \If {$2^{x+y} \le \tilde{\out}$}
    \State \textbf{for each} $(I,J) \in \mathcal{D}$ of type $(x,y)$: Move $(I,J)$ from $\mathcal{D}$ to $\mathcal{C}$
  \Else
    \If{ $\mathcal{D}$ contains more than $q$ subproblems of type $(x,y)$}
      \State Compute $A_I+B_J$ in parallel for all $(I,J) \in \mathcal{D}$ of type $(x,y)$
      \State Stop once all but $q$ of these calls finished
      \State \textbf{for each} finished call $A_I+B_J$: Move $(I,J)$ from $\mathcal{D}$ to $\mathcal{C}$
    \EndIf
    \For{\textbf{each} remaining $(I,J) \in \mathcal{D}$ of type $(x,y)$}
      \State Split $I = [i_1,i_2]$ at $i = \lfloor (i_1+i_2)/2 \rfloor$ into $I_1$ and $I_2$
      \State Determine the maximum $j \in J$ with $A_i+B_j \le u$
      \State Split $J$ at $j$ into $J_1$ and $J_2$
      \State Add $(I_1,J_1)$ to $\mathcal{C}$
      \State Add $(I_1,J_2), (I_2,J_1)$ to $\mathcal{D}$
      \State Remove $(I,J)$ from $\mathcal{D}$
    \EndFor
  \EndIf
\EndWhile
\State \Return $\mathcal{C}$
\EndProcedure
\end{algorithmic}
\end{algorithm}

\subsection{Analysis}

The correctness of the algorithm, meaning that the output is a unique rectangle covering, follows from the next claim and the fact that the algorithm stops when $\mathcal{D} = \emptyset$.

\begin{claim}[Invariant that $\mathcal{C} \cup \mathcal{D}$ is a covering]
\label{lem:inv_covering}
  At any point during the algorithm, the family $\mathcal{C} \cup \mathcal{D}$ is a unique rectangle covering of $(A,B,[u])$. 
\end{claim}
\begin{proof}
  Moving $(I,J)$ from $\mathcal{D}$ to $\mathcal{C}$ does not change this property. Therefore, the only crucial step is the splitting of $I$ into $I_1,I_2$ and of $J$ into $J_1,J_2$, where we add $(I_1,J_1)$ to $\mathcal{C}$ and add $(I_1,J_2), (I_2,J_1)$ to~$\mathcal{D}$. Observe that at this point we have $u < A_{i}+B_{j+1} \le A_{i+1}+B_{j+1}$, and thus the rectangle $(I_2,J_2)$ does not contain any sum below $u$, so it is unnecessary for a covering. It follows that $\mathcal{C} \cup \mathcal{D}$ remains a covering.
  Moreover, noting that $I_1,I_2,J_1,J_2$ are again intervals, it remains rectangular. Finally, since no pair $(i,j)$ is contained in two subproblems $I_b \times J_{b'}$, it remains unique. 
\end{proof}

We will need the following claims to analyze the running time and the cost of the covering.

\begin{claim}[Subproblems of type $(x,y)$ form a staircase] 
\label{lem:staircase}
For any distinct subproblems $(I,J), (I',J')$ of the same type $(x,y)$, we have $\max(I') < \min(I)$ and $\max(J) < \min(J')$ (or vice versa).
Moreover, if this holds then $A_{I'} + B_J \subseteq [u]$.
\end{claim}
\begin{proof}
This is true in the beginning since $|\mathcal{D}|=1$. Moving $(I,J)$ from $\mathcal{D}$ to $\mathcal{C}$ cannot violate this property.
Therefore, the only crucial step is the splitting of $I$ into $I_1,I_2$ and of $J$ into $J_1,J_2$, where we remove $(I,J)$ from $\mathcal{D}$ and add $(I_1,J_2), (I_2,J_1)$ to~$\mathcal{D}$. In this situation, clearly $(I_1,J_2), (I_2,J_1)$ form a staircase, and $A_{I_1} + B_{J_1} \subseteq [u]$. Thus, the property is maintained ``locally''. It is not hard to see that the property is also maintained ``globally'', when we have distinct subproblems $(I,J), (I',J')$ that satisfy the property and we split both of them into subproblems.
\end{proof}

\begin{claim}[Any subproblem creates at most two subproblems of strictly smaller type]
Processing a subproblem $(I,J) \in \mathcal{D}$ can cause the insertion of at most two new subproblems into $\mathcal{D}$, both of strictly smaller type.
\end{claim}
\begin{proof}
The only point at which we add new subproblems to $\mathcal{D}$ is the splitting phase. So let $(I,J)$ be a subproblem of type $(x,y)$ and consider the splitting of $I$ into $I_1,I_2$ and of $J$ into $J_1,J_2$, where we remove $(I,J)$ from $\mathcal{D}$ and add $(I_1,J_2), (I_2,J_1)$ to~$\mathcal{D}$. Since we split $I$ at the midpoint, we have $|I_r| \le \lceil |I|/2 \rceil$, for any $r \in \{1,2\}$. We clearly also have $|J_r| \le |J|$. Hence, the new type $(x_r,y_r) = (\lceil \log |I_r| \rceil, \lceil \log |J_r| \rceil)$ satisfies\footnote{To be precise, here we use that for any integers $z,w$ the inequalities $2^{w-1} < z \le 2^w$ imply $2^{w-2} < \lceil z/2 \rceil \le 2^{w-1}$, and thus $\lceil \log (\lceil z/2 \rceil) \rceil \le \lceil \log z \rceil - 1$.} $x_r < x$ and $y_r \le y$. As this implies $x_r+y_r < x+y$, the newly added subproblems have a strictly smaller type.
\end{proof}

\begin{claim}[Invariant on the number of subproblems] \label{lem:number_of_subproblems}
At any point during the algorithm, there are at most $2 q \log(n) \log(m)$ subproblems of any fixed type $(x,y)$.
\end{claim}
\begin{proof}
The claim is immediate for $(x,y)=( \lceil \log n \rceil, \lceil \log m \rceil )$. Fix a type $(x,y)$ and assume that the claim is true for every type $(x',y')$ with $(x,y) \prec (x',y')$. 
Note that for any type $(x',y')$ at most $q$ subproblems of type $(x',y')$ reached the splitting phase, and each such subproblem gave rise to at most two newly added subproblems. Since the number of different types is at most $\log(n) \log (m)$, we obtain the claimed bound on the number of subproblems of type $(x,y)$.
\end{proof}


The following claim lies at the core of the analysis of our covering construction, as it bounds the running time and added cost of Case $2$.

\begin{claim} \label{claim:charging}
Fix a type $(x,y)$ with $2^{x+y} > \tilde{\out}$. Lines 11-13 of Algorithm~\ref{alg:covering_construction} take time $\tOh( \out^{4/3})$ and add rectangles of total cost $\tOh( \out^{4/3})$ to $\mathcal{C}$.
\end{claim}
\begin{proof}
We denote by $(I_1, J_1), (I_2, J_2), \ldots, (I_R, J_R)$ the subproblems of type $(x,y)$. 
We can assume $R \ge q$, since otherwise lines 11-13 are not called.
Note that $R \le 2 q \log(n) \log(m)$ by Claim~\ref{lem:number_of_subproblems}.
By Claim~\ref{lem:staircase} these subproblems form a staircase, so we can assume that 
\begin{align*}
  \max(I_1) < \min(I_2), &\max(I_2) < \min(I_3), \ldots, \max(I_{R-1}) < \min(I_R) \\
  \min(J_1) > \max(J_2), &\min(J_2) > \max(J_3), \ldots, \min(J_{R-1}) > \max(J_R).
\end{align*}
Claim~\ref{lem:staircase} also implies that for any $r < \ell$ we have
\begin{align} \label{eq:myniceeq} 
  A_{I_r} + B_{J_\ell} \subseteq [u].
\end{align}

Set $\delta := 1 / (6 \log(n) \log(m))$, so that $3 \delta R \le q$. 
For any $r \in [R]$ set
\[ S(r) := \sum_{i=1}^{r-1} \sum_{j=r+1}^R |A_{I_r} + B_{J_j}| + |A_{I_i} + B_{J_j}| + |A_{I_i}+B_{J_r}|. \]
\begin{claim} \label{cla:myinterestingclaim}
  There are at least $(1-\delta)R$ indices $r$ with $S(r) \le 6 R \, \out / \delta$.
\end{claim}
\begin{proof}
  By simple counting how often a summand $|A_{I_r} + B_{J_\ell}|$ can appear, we observe that 
  \[ \sum_{r=1}^R S(r) \le 3 R \sum_{r < \ell} |A_{I_r} + B_{J_\ell}|. \]
  Note that we have $|A_{I_r} + B_{J_\ell}| \le \out$ by equation (\ref{eq:myniceeq}).
  However, we need a stronger property.
  We decompose the sum $\sum_{r < \ell} |A_{I_r} + B_{J_\ell}|$ into $2R-1$ diagonal sums of the form $\sum_r |A_{I_r} + B_{J_{r+\Delta}}|$. Since $\max(A_{I_r} + B_{J_\ell}) < \min(A_{I_{r+1}} + B_{J_{\ell+1}})$, all summands in a diagonal sum are disjoint, and thus each diagonal sum is bounded from above by $\out$. We thus obtain
  \[ \sum_{r < \ell} |A_{I_r} + B_{J_\ell}| \le (2R-1) \out, \]
  which yields
  \[ \sum_{r=1}^R S(r) \le 6 R^2 \out. \]
  By Markov's inequality, it follows that all but $\delta R$ indices $r$ satisfy $S(r) \le 6 R\, \out / \delta$. 
\end{proof}

In the remainder we consider only indices $\delta R \le r \le (1-\delta)R$ with $S(r) \le 6 R \, \out / \delta$; note that there are at least $(1-3\delta) R$ such indices $r$. For any such $r$, there are at least $\frac {\delta R^2}2$ pairs $(i,j)$ with $1 \le i < r$ and $r < j \le R$. Hence, there exist $i,j$ with $i < r < j$ and
\[ |A_{I_r} + B_{J_j}| + |A_{I_i} + B_{J_j}| + |A_{I_i}+B_{J_r}| \le \Big( \frac{6 R \, \out}\delta \Big) \Big/ \Big( \frac {\delta R^2}2 \Big) = \frac{12 \, \out}{\delta^2 R}. \]
We continue by invoking the corollary of Ruzsa's triangle inequality (Lemma~\ref{lem:ruzsa}), see Figure~\ref{fig:fig1} for an illustration:
	\begin{align*}
		|A_{I_r} + B_{J_r}| \leq
	\frac{ |A_{I_r}  + B_{J_j }| \cdot |B_{J_j} + A_{I_i}  | \cdot |A_{I_i} + B_{J_{r}}|  }{ |A_{I_i}| \cdot | B_{J_j}| } \leq \\
\frac{12^3 \out^3 / (\delta^6 R^3)}{2^{x-1} \cdot 2^{y-1}}.
\end{align*}
	By the case assumption $2^{x+y} > \tilde{\out} \ge \out$ and $R \ge q = \lceil \tilde{\out}^{1/3} \rceil \ge \out^{1/3}$, and by our choice of $\delta = 1 / (6 \log(n) \log(m))$, we obtain 
	\[ |A_{I_r} + B_{J_r}| \leq \Oh(\out \cdot \log^{12}(nm)). \]
	Recall that this holds for at least $(1-3\delta)R \ge R - q$ many indices $r$. 
	
	Since we compute all sumsets $A_{I_r} + B_{J_r}$ in parallel and stop once all but $q$ calls are finished, it follows that we stop after spending time $\tOh(\out)$ for each call. Over $R \le 2 q \log(n) \log (m) = \tOh(\out^{1/3})$ calls, this takes total time $\tOh(\out^{4/3})$. Moreover, any finished call ran in time $\tOh(\out)$, so it has output-size $\tOh(\out)$, so it contributes cost $\tOh(\out)$. Thus, we add rectangles of total cost $\tOh(\out^{4/3})$ in every invocation of lines 11-13 of Algorithm~\ref{alg:covering_construction}. 
\end{proof}

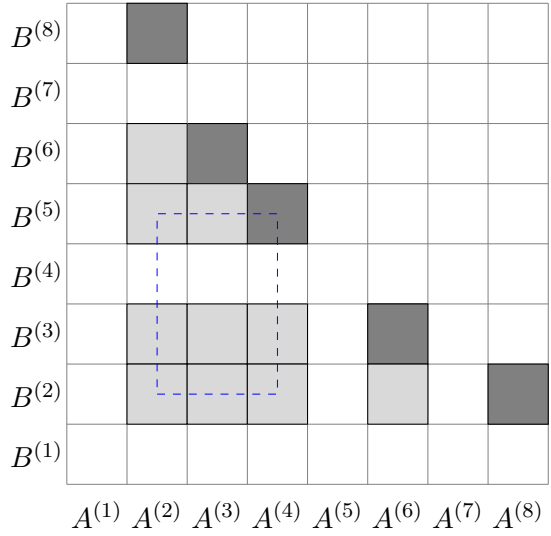
\begin{figure}
\begin{center}
\begin{tikzpicture}[scale=0.8]
\foreach \i in {\xMin,...,\xMax} {
        \draw [very thin,gray] (\i,\yMin) -- (\i,\yMax); 
    }
    \foreach \i in {\yMin,...,\yMax} {
        \draw [very thin,gray] (\xMin,\i) -- (\xMax,\i); 
    }
\foreach \i in {1,...,8}
{
	\draw node at  (\i+0.5,\yMin-0.5) {$A^{(\i)}$};
	\draw node at (\xMin-0.5,\i+0.5) {$B^{(\i)}$};
}
\def\farray{{
2,3,4,6,8
}};

\def\uarray{{
8,6,5,3,2
}};

\foreach \i in {0,...,4}
{
	\draw [fill=gray] (\farray[\i],\uarray[\i]) rectangle (\farray[\i]+1,\uarray[\i]+1)  ;
}

\foreach \i in {0,...,4}
\foreach \j in {0,...,4}
	{
		\ifnum\j>\i
			\draw [fill=gray!30] (\farray[\i],\uarray[\j]) rectangle (\farray[\i]+1,\uarray[\j]+1) ;
		\fi
	}
\draw [dashed, very thin, blue] (\farray[2] -1.5, \uarray[2] - 2.5 ) rectangle (\farray[2] +0.5,\uarray[2] +0.5);

\end{tikzpicture}
\end{center}
\caption{Illustration of the proof of Claim~\ref{cla:myinterestingclaim}. The thick gray squares depict the unprocessed subproblems $(I_1,J_1),(I_2,J_2), \ldots$  of a specific type $(x,y)$. The light gray squares depict pairs $(I_r,J_\ell)$ for $r < \ell$. The vertices of the blue rectangle correspond to (starting from the top-right corner and going clockwise) $(I_r, J_r), (I_r, J_j), (I_i, J_j),(I_i, J_r)$. From such a 4-tuple, where all sumsets $|A_{I_r} + B_{J_j}|, |A_{I_i} + B_{J_j}|,|A_{I_r} + B_{J_j}|$ are ``small'', Ruzsa's triangle inequality (Lemma~\ref{lem:ruzsa}) allows us to bound $|A_{I_r}+ B_{J_r}|$ from above. }
\label{fig:fig1}
\end{figure}

\begin{claim}[Cost Bound]
  For any type $(x,y)$ we add rectangles of total cost $\tilde{O}( \out^{4/3} )$ to $\mathcal{C}$. 
\end{claim}
\begin{proof}
  If $2^{x+y} > \tilde{\out}$, then the cost contributed by lines 11-13 of Algorithm~\ref{alg:covering_construction} is bounded in Claim~\ref{claim:charging}. Additionally, in the splitting phase (line 18) we add a rectangle $(I_1,J_1)$ satisfying $A_{I_1} + B_{J_1} \subseteq [u]$, and thus $|A_{I_1} + B_{J_1}| \le \out$. Since we split at most $q$ subproblems, we add a total cost of $\tOh(q \cdot \out) = \tOh(\out^{4/3})$. 
  
  If $2^{x+y} \le \tilde{\out}$, then by the trivial bound each added rectangle $(I,J)$ has cost at most $|I| \cdot |J| \le 2^{x+y} \le \tilde{\out} = O(\out)$. By Claim~\ref{lem:number_of_subproblems} we have $\tOh(q)$ subproblems, and thus we add a total cost of $\tOh(q \cdot \out) = \tOh(\out^{4/3})$. 
\end{proof}

Over all types $(x,y)$, we obtain a total cost of $\tOh(\out^{4/3})$. The running time bound of $\tOh(\out^{4/3})$ follows along the same lines.
This finishes the proof of Theorem~\ref{thm:technical_core}.

%% file: covering_lb.tex

\section{Lower Bound on Coverings} \label{sec:covering_lb}

This section is devoted to proving Theorem \ref{thm:covering_lb}. 
We will use the following notation. For a vector~$x$ we write $w(x)$ for its Hamming weight, i.e., its number of non-zero coordinates.
For vectors $x,y$ we write $d_H(x,y)$ for their Hamming distance, i.e., the number of coordinates in which they differ.

For sets $X,Y$ we denote their symmetric difference by $X \triangle Y := (X \setminus Y) \cup (Y \setminus X)$.
We identify sets with their indicator vectors. In particular, note that if $x,y$ are the indicator vectors of sets $X,Y$, respectively, then $d_H(x,y) = |X \triangle Y|$.

Recall  the definition of the binary entropy function $h(x) = -x \log x - (1-x) \log (1-x)$ and note that $h(1/2) = 1$.
Finally, we use notation $\Oh_\delta(.)$, $\Omega_\delta(.)$, and $\Theta_\delta(.)$ to hide constants that only depend on the parameter $\delta$.

The following is a standard construction of an error correcting code.

\begin{lemma}[Constant-weight Binary Code]\label{lem:code}
Fix $0 \le \delta < 1/2$. For any even integer $t$ there exists a code $ \mathcal{E} \subseteq \{0,1\}^t$ such that:
\begin{itemize}
\item Each codeword $x \in \mathcal{E}$ has weight $w(x) = t/2$,
\item Any two codewords $x,y \in \mathcal{E}$ with $x \neq y$ satisfy $d_H(x,y) \geq \delta t$, and
\item The number of codewords is $|\mathcal{E}| = \Omega_\delta(2^{(1-h(\delta))t})$.
\end{itemize}
\end{lemma}

\begin{proof}
We pick $\mathcal{E}$ greedily among all ${t \choose t/2}$ many vectors $x \in \{0,1\}^t$ with Hamming weight~$t/2$. Whenever we pick a vector $x$, we mark all vectors $y$ within distance $\delta t$ of $x$ as unpickable. 
Note that we mark at most $\sum_{i=0}^{\delta t} {t \choose i}$ many vectors $y$. In total, this process picks $|\mathcal{E}| \ge {t \choose t/2} \big/ \big(\sum_{i=0}^{\delta t} {t \choose i} \big)$ vectors. 
The claim now follows from the following facts about binomial coefficients:
\begin{itemize}
  \item $\sum_{i=0}^{\alpha n} {n \choose i} = \Theta_\alpha \big({n \choose \alpha n}\big)$ for any $0 \le \alpha < 1/2$,
  \item ${n \choose  \alpha n} = \Theta_\alpha \big(2^{h(\alpha)n}/\sqrt{n}\big)$.
\end{itemize}
Indeed, with these facts we obtain $|\mathcal{E}| \ge {t \choose t/2} \big/ \big(\sum_{i=0}^{\delta t} {t \choose i} \big) = \Omega_\delta\big( {t \choose t/2} \big/ {t \choose \delta t} \big) = \Omega_\delta\big( 2^{(1-h(\delta))t} \big)$.
\end{proof}

We next lift the above code to a family of sets where the sumset $X^{(i)} + Y^{(j)}$ has large cardinality if $i=j$, and small cardinality if $i \ne j$.

\begin{lemma} \label{lemma:covering_lb}
Fix $0 \le \delta < 1/2$. For any integers $m,t \ge 2$, where $t$ is even, for 
\[ \sigma := m^t, \quad \alpha := 2^{\delta t/2} m^{(1-\delta/2)t}, \quad \text{and some} \quad g = \Omega_\delta\big(2^{(1-h(\delta))t}\big), \]
there exist sets $X^{(1)},..,X^{(g)}, Y^{(1)},..,Y^{(g)} \subseteq [\sigma]$ satisfying
\begin{enumerate}
\item $|X^{(i)}| = |Y^{(i)}| = \sigma^{1/2}$ for any $1 \le i \le g$,
\item $|X^{(i)}+Y^{(i)}| = \sigma$ for any $1 \le i \le g$, and
\item $|X^{(i)} + Y^{(j)}| \leq \alpha$ for any $i \neq j$.
\end{enumerate}
\end{lemma}

\begin{proof}
Let $\mathcal{E} \subseteq \{0,1\}^t$ be the code given by Lemma~\ref{lem:code} and set $g := |\mathcal{E}|$.
Let $I^{(1)},\ldots,I^{(g)} \subseteq \{0,1,\ldots,t-1\}$ be sets such that the codewords in $\mathcal{E}$ are the indicator vectors of $I^{(1)},\ldots,I^{(g)}$.
For any $1\le \ell \le g$, we write 
\[ \overline{I}^{(\ell)} := \{0,\ldots,t-1\} \setminus I^{(\ell)}. \] 
Moreover, for any $I \subseteq \{0,..,t-1\}$ we set
\[	S(I) := \Big\{	\sum_{i \in I} \eta_i \cdot m^i \;\Big\vert\; 0 \le \eta_i < m \Big\},	\]
where $\eta_i$ ranges over all integers between 0 and $m-1$.
Note that $S(I)$ is the set of all numbers with $t$ digits in the $m$-ary system for which every digit outside of $I$ is $0$. Alternatively, it can be viewed as the sumset of $|I|$ many arithmetic progressions of length $m$ and step sizes $\{m^i\}_{i \in I}$. 
Finally, we set 
\[ X^{(\ell)} := S(I^{(\ell)}) \qquad \text{and} \qquad Y^{(\ell)} = S(\overline{I}^{(\ell)}). \]

\medskip
Note that any number in $S(I)$ is bounded from above by $\sum_{i=0}^{t-1} (m-1) \cdot m^i = m^t-1$, and thus we have $X^{(\ell)},Y^{(\ell)} \subseteq [m^t]$ for any $\ell$.
It remains to verify the three claims.

For (1.), note that for any $\ell$ we have $|X^{(\ell)}| = |S(I^{(\ell)})| = m^{|I^{(\ell)}|} = m^{t/2}$, since the code $\mathcal{E}$ has constant weight~$t/2$. The same holds for $Y^{(\ell)}$.

For the remaining claims, for any $\ell,h$ we write
\begin{align}
\big|X^{(\ell)}+Y^{(h)}\big| &\;=\; \big|S(I^{(\ell)}) + S(\overline{I}^{(h)})\big|  
\;=\; \bigg| \bigg\{ \sum_{i \in I^{(\ell)} } \eta_i \cdot m^i  + \sum_{j \in \overline{I}^{(h)} } \eta'_j \cdot m^j \;\bigg\vert\; 0 \le \eta_i, \eta'_j < m \bigg\} \bigg|. \label{eq:wohfdf}
\end{align}

For (2.), we use that $I^{(\ell)} \cup \overline{I}^{(\ell)}$ is a partitioning of $\{0,\ldots,t-1\}$ to obtain for any $\ell$
\[
\big|X^{(\ell)}+Y^{(\ell)}\big| \;=\; \bigg| \bigg\{ \sum_{i=0}^{t-1} \eta_i \cdot m^i \;\bigg\vert\; 0 \le \eta_i < m \bigg\} \bigg|  
 \;=\; m^t.
\]

For (3.), for any $\ell \ne h$ we write $I := I^{(\ell)}$ and $J := \overline{I}^{(h)}$ and express the right hand side of (\ref{eq:wohfdf}) in terms of the intersection $I \cap J$ and the symmetric difference $I \triangle J$ as
\begin{align*}
\big|X^{(\ell)}+Y^{(h)}\big| &\;=\; \bigg| \bigg\{ \sum_{i \in I \cap J } \eta_i \cdot m^i  + \sum_{j \in I \triangle J } \eta'_j \cdot m^j \;\bigg\vert\; 0 \le \eta_i \le 2 m - 2,\, 0 \le \eta'_j < m \bigg\} \bigg|.
\end{align*}

This allows us to bound
\[ \big|X^{(\ell)}+Y^{(h)}\big| \;\le\; (2m-1)^{|I \cap J|} \cdot m^{|I \triangle J|} \;\le\; 2^{|I \cap J|} \cdot m^{|I \cap J|+|I \triangle J|} \;=\; 2^{|I \cap J|} \cdot m^{|I \cup J|} \;=\; 2^t \Big( \frac m2 \Big)^{|I \cup J|}, \]
where we have used inclusion-exclusion $|I \cup J| = |I|+|J|-|I \cap J|$ combined with $|I|=|J|=t/2$ in the last step.
We now use the fact that for any $S,T \subseteq \{0,\ldots,t-1\}$ we have
\[ \big| S \cup \big(\{0,\ldots,t-1\} \setminus T\big) \big| = \frac 12 \big( 2t + |S| - |T| - |S \triangle T| \big). \]
Plugging in the bounds of $|I^{(\ell)}| = t/2$ and $|I^{(\ell)} \triangle I^{(h)}| \ge \delta t$ for any $\ell \ne h$ by the properties of the code $\mathcal{E}$, we obtain the bound
\[ |I \cup J| \;=\; \big| I^{(\ell)} \cup \overline{I}^{(h)} \big| \;\le\; \Big(1 - \frac \delta 2 \Big) t. \]
Together, this yields
\[ \big|X^{(\ell)}+Y^{(h)}\big| \;\le\; 2^t \Big( \frac m2 \Big)^{(1-\delta/2)t} = 2^{\delta t/2} m^{(1-\delta/2)t}, \]
finishing the proof.

\end{proof}


\begin{lemma} \label{lem:lb_setsAB}
  With parameters $\sigma,\alpha,g$ as in Lemma~\ref{lemma:covering_lb}, there exist sets $A,B \subseteq \mathbb{N}$ and an integer~$u$ such that 
  \begin{enumerate}
    \item $|A|,|B| = g \cdot \sigma^{1/2}$, 
    \item $\out := |(A+B) \cap [u]| = \Oh(g^2 \cdot \alpha + \sigma)$, and
    \item Any rectangle covering of $(A,B,[u])$ has cost $\Omega(g \cdot \sigma)$.
  \end{enumerate}
\end{lemma}
\begin{proof}
  Let $X^{(1)},..,X^{(g)}, Y^{(1)},..,Y^{(g)}$ be the sets constructed in Lemma~\ref{lemma:covering_lb}. 
  For $i \in \mathbb{N}$ we let
  \[ E(i) := \begin{cases} i, & \text{if $i$ is even} \\0, & \text{otherwise.} \end{cases} \] 
  Moreover, we let $M$ be a sufficiently large integer; setting $M := 100 (\sigma+g)$ suffices.
  With this setup, for any $1 \le i,j \le g$ we define $A^{(i)}$ and $B^{(j)}$ as appropriate shifts of $X^{(i)}$ and $Y^{(j)}$, respectively, more precisely we set
  \begin{align*}
    A^{(i)} := X^{(i)} + \big\{ i \cdot M^2 + E(i) \cdot M \big\} \qquad &\text{and} \qquad A := \bigcup_{i=1}^g A^{(i)},  \\
    B^{(j)} := Y^{(j)} + \big\{ (g-j) \cdot M^2 \big\} \qquad &\text{and} \qquad B := \bigcup_{j=1}^g B^{(j)}.
  \end{align*}
Finally, we set
\[ u := g \cdot M^2 + 2\sigma. \]

\medskip
We now verify the three claims. The size bound $|A|,|B| = g \cdot \sigma^{1/2}$ is immediate from the property $|X^{(i)}|,|Y^{(j)}| = \sigma^{1/2}$.

For the output-size, we use the following claim.
\begin{claim}
  The following properties hold.
  \begin{enumerate}
    \item For any $i \ne j$ we have $|A^{(i)} + B^{(j)}| \le \alpha$,
    \item For any even $i$ we have $(A^{(i)} + B^{(i)}) \cap [u] = \emptyset$, and
    \item For any odd $i$ we have $A^{(i)} + B^{(i)} \subseteq [g \cdot M^2, g\cdot M^2 + 2\sigma] \subseteq [u]$.
  \end{enumerate}
\end{claim}
\begin{proof}
  Claim 1.\ is immediate from Lemma~\ref{lemma:covering_lb}.3, since $A^{(i)}$ and $B^{(j)}$ are just shifts of $X^{(i)}$ and $Y^{(j)}$.
  
  For even $i$ we have $\min(A^{(i)}) \ge i \cdot M^2 + M$ and $\min(B^{(i)}) \ge (g-i)\cdot M^2$ and thus $\min(A^{(i)} + B^{(i)}) \ge g \cdot M^2 + M > u$, which shows claim 2.
  
  For odd $i$ we have $\min(A^{(i)}) \ge i \cdot M^2$ and thus $\min(A^{(i)} + B^{(i)}) \ge g \cdot M^2$. Similarly, we have $\max(A^{(i)}) \le i \cdot M^2 + \sigma$ and $\max(B^{(j)}) \le (g-i)\cdot M^2 + \sigma$ and thus $\max(A^{(i)} + B^{(i)}) \le g \cdot M^2 + 2 \sigma$.
\end{proof}
The above claim allows us to bound
\[
  \out \;=\; \big( A \cap B \big) \cap [u] \;\le\; \bigg( \sum_{i \ne j} |A^{(i)} + B^{(j)}| \bigg) + \bigg| \bigcup_{\text{odd }i} \big( A^{(i)} + B^{(i)} \big) \cap [u] \bigg| 
  \;\le\; g^2 \cdot \alpha + 2\sigma + 1.
\]

It remains to analyze coverings. So let $\mathcal{C}$ be a rectangle covering of $(A,B,[u])$. See Figure~\ref{fig:lb_covering} for an illustration.
We construct a graph $G$ with vertex set $V(G)$ consisting of all odd integers $1 \le i \le g$. We put an edge $(i,i+2)$ into $E(G)$ if there exists a rectangle in $\mathcal{C}$ that contains a pair in $A^{(i)} \times B^{(i)}$ as well as a pair in $A^{(i+2)} \times B^{(i+2)}$. Note that such a rectangle contains \emph{all} pairs in $A^{(i+1)} \times B^{(i+1)}$ and thus has cost at least 
\[ \big| A^{(i+1)} + B^{(i+1)} \big| = \sigma. \]
We now consider two cases, depending on the number of edges in $G$. 

If $G$ has more than $g/4$ edges, then for at least $g/4$ even integers $i$ all pairs in $A^{(i)} \times B^{(i)}$ are covered by $\mathcal{C}$. Since $A^{(i)} + B^{(i)} \subseteq [g \cdot M^2 + i \cdot M, g \cdot M^2 + i \cdot M + 2\sigma]$ has size $\sigma$, and $M$ is large, all sumsets $A^{(i)} + B^{(i)}$ for even $i$ are disjoint, and therefore any covering of $g/4$ such sumsets has cost at least $g \sigma /4$.

Otherwise, if $G$ has at most $g/4$ edges, then it has at least $g/4$ components. Note that each component must be covered by distinct rectangles in $\mathcal{C}$. Since each component requires cost at least $|A^{(i)} + B^{(i)}| = \sigma$, the covering has a total cost of at least $g \sigma / 4$. This finishes the proof.
\end{proof}

We are now ready to prove Theorem \ref{thm:covering_lb}.

\begin{proof}
Lemma~\ref{lem:lb_setsAB} yields for any $\delta > 0$ and integers $m,t$, where $t$ is even, a tuple $(A,B,[u])$ with
\[ \out \;:=\; \big| (A+B) \cap [u] \big| \;=\; \Oh_\delta\big( 2^{(2 - 2h(\delta) + \delta/2)t} m^{(1-\delta/2)t} + m^t \big), \]
and any rectangle covering of $(A,B,[u])$ has cost $\Omega_\delta( 2^{(1-h(\delta))t} m^t )$.
  Observe that we can write this cost bound in the form $\Omega_\delta (\out^c)$ for 
  \[ c := \frac{ (1 - h(\delta))t + t \log m }{ \max\{ (2-2h(\delta)+\delta/2)t + (1-\delta/2) t \log m,\, t \log m\} }. \]
  Note that $t$ cancels in this expression. We numerically optimize $c = c(\delta,m)$ by setting $m := 10$ and $\delta := 0.2709$, obtaining $c \ge 1.047$.
  In particular, we constructed an infinite sequence of tuples $(A,B,[u])$ for which any rectangle covering has cost $\Omega(\out^{1.047})$, which proves Theorem~\ref{thm:covering_lb}.
\end{proof}

\begin{figure}
\begin{center}
\begin{tikzpicture}

\def\foo{{
0,1,2,3,4,5,6,7,8,9,
}};

\foreach \i in {1,...,9} {
        \draw [very thin,gray] (\i,\yMin) -- (\i,\xMax+1);
    }
\foreach \i in {1,...,9} {
        \draw [very thin,gray,] (\xMin,\i) -- (\xMax+1,\i); 
    }
\foreach \i in {1,...,9}{
	\draw node at  (\i+0.5,\yMin-0.5) {$A^{(\i)}$};
}
\foreach \i in {1,...,9}{
	\draw node at (\xMin-0.5,10-\i+0.5) {$B^{(\i)}$};
}

\foreach \i in {0,...,8} {
	\foreach \j in {0,...,8} {
		\ifnum \j < \i
				\draw [fill=gray!50] (9-\foo[\i],\foo[\j]+1) rectangle (10 -\foo[\i] ,\foo[\j] +2);
		\fi
	}
}
	\foreach \i in {0,2,4,6,8} \draw [fill=gray!50] (9-\foo[\i],\foo[\i]+1) rectangle (10-\foo[\i],\foo[\i]+2);
	\foreach \i in {0,2,4,6,8} \draw node[red,scale=1] at  (\xMin+ \foo[\i] +0.45,\yMax - \foo[\i] + 0.5) { $0$ };
	\foreach \i in {2,4,6,8} \draw node[red,scale=0.8] at  (\xMin+ \foo[\i-1] +0.45,\yMax - \foo[\i-1] + 0.5) { $\i$ };
	\draw [thick, blue] (\xMin+2, \yMin) rectangle (\xMin + 6.75, \yMin+5.25);

\end{tikzpicture}
\end{center}
\caption{A tuple $(A,B,[u])$ constructed in the proof of Theorem \ref{thm:covering_lb}. In gray we mark all the pairs $(i,j)$ such that $A^{(i)} + B^{(j)}$ is contained in $[u]$. On the diagonal only every second pair belongs to the output. The numbers along the diagonal indicate the shift $E(i)$. The five boxes marked as ``0'' contribute to $(A+B)\cap [u]$ exactly the same numbers.  The rectangles marked with $2,4,6$, and $8$ have disjoint sumsets. By definition, a covering algorithm should cover all the boxes marked with~$0$. The indicated blue rectangle covers numbers from two boxes marked by $0$, and thus also fully contains the box corresponding to pair $(A^{(6)},B^{(6)})$ in between, which does not belong to the output.}
\label{fig:lb_covering}
\end{figure}
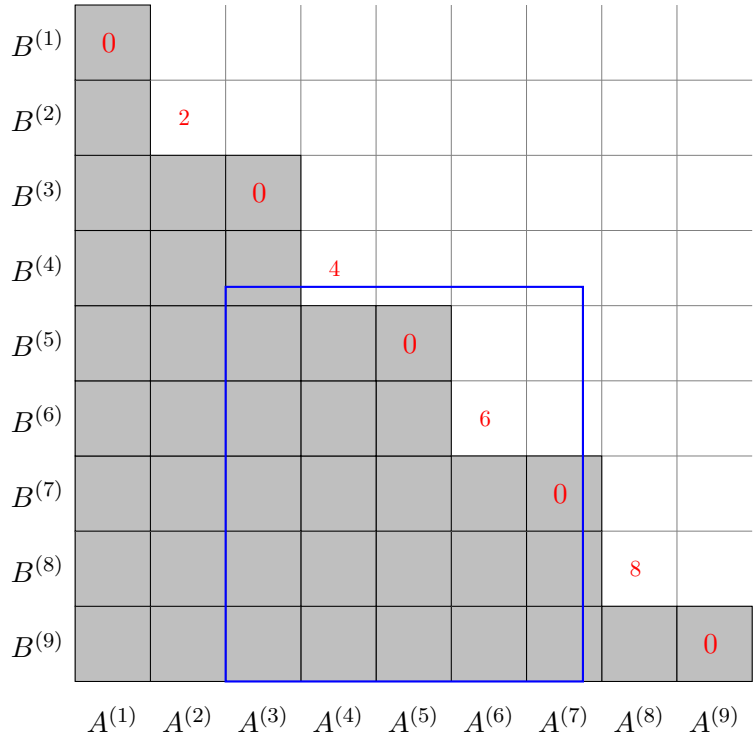

%% file: reduction.tex
\section{Reducing \textsc{SubsetSum} to Prefix-Restricted Sumset Computation} \label{sec:subset_sum}

This section is devoted to proving Theorems~\ref{thm:subsetsum_result} and~\ref{thm:add_relaxed_upper_bound}. In particular, we give an output-sensitivity-preserving reduction from \textsc{SubsetSum} to top-$k$-convolution. For that, we need the following definition.
\begin{definition}[$(\alpha,\zeta)$-effective algorithm] 
Let sets $A,B \subseteq [u]$ with $|A|=n,|B|=m$.
An $(\alpha,\zeta)$-effective algorithm is an algorithm for prefix-restricted sumset computation algorithm on instance $(A,B,[u])$ which runs in time 
\[\tilde{O}(m +n  + |(A+B)\cap [(1+\zeta)u]|^{1+\alpha}).\]
\end{definition}

The reduction is based on randomly dividing the input and conquering with a prefix-restricted sumset computation. In order to prove correctness, apart from the argumentation in~\cite{Bring17},  we additionally we need number-theoretic condition on how the set of attainable subset sums decreases under partition; this is captured in Lemma~\ref{lem:splitting_lemma}. The base cases of the algorithm are those instances where either there is only one element or every element is sufficiently large with respect to the target. 

\subsection{Handling Large Elements}

In this subsection we treat one of the base instances of the more general algorithm, which are the instances where all elements are large with respect to the target. For a small technical reason in later subsections (in particular, in order to afford to take a union bound over all recursive calls), we need to define small numbers with respect to two parameters $u,t$. We shall analyze Algorithm~\ref{alg:large_elements}.  We classify an element as ``heavy'' if it is larger than $\frac{u}{\beta \log^3 t}$.

\begin{algorithm}[!hb]\caption{}\label{alg:large_elements}
\begin{algorithmic}[1]
\Procedure{\textsc{SubsetSumforLargeElements}}{$X,u,t,\delta$}  \Comment{$x \in \left[\frac{u}{ \beta\log^3 t},u\right], \forall x \in X; \beta = \Theta(1)$}
\State $O \leftarrow \emptyset$					
\State $R \leftarrow \Theta( \log(t/\delta) )$
\State $B \leftarrow 2\beta^2 \log^6 t$
\For { $ r \in [R]$}
	\State Color $X$ randomly with $B$ different colors.	\label{lin:color_coding_large}
	\For {$b \in [B]$}
	\State $X^{(r,b)} \leftarrow $ (elements of $X$ which have received color $b$)$\bigcup \{0\}$.
	\EndFor
	\State $O^{(r)} \leftarrow \emptyset$
	\For {$b \in [B]$} \label{lin:large_sumset1}
	\State $O^{(r)} \leftarrow \left( O^{(r)} + X^{(r,b)} ) \right)\bigcap [u]$ \Comment{Oracle call with failure probability $\frac{\delta}{2 B R}$} \label{lin:boost}
	\EndFor\label{lin:large_sumset2}
	\State $ O \leftarrow O \cup O^{(r)}$
\EndFor
\State Return $O$
\EndProcedure
\end{algorithmic}
\end{algorithm}

The instance consisting solely of large items is much easier to solve, since only a polylogarithmic number of elements can participate in a subset sum which is at most $t$. Our algorithm in the next subsection will be recursive, and the next lemma will serve as one of the the bases of the recursion. 
 
\begin{lemma}[Guarantee for large elements]\label{lem:large_items}
Suppose there exists an $(\alpha,\zeta)$-effective algorithm for solving $(A,B,u)$ prefix-resticted sumset computation with probability $9/10$. Let parameters $u\leq t$. Let $X \subseteq \left[\frac{u}{ \beta\log^3 t},u\right]$ a set of positive integers, where $\beta$ is a sufficiently large absolute constant. We can find $\mathcal{S}(X,u)$ using Algorithm \ref{alg:large_elements} in time 
\[	 \tilde{O} \left( |\mathcal{S}(X,u+\zeta \cdot u)|^{1+\alpha} \right), \]
 with probability $1-\delta$. We remind the reader that $\tilde{O}(\cdot)$ hides factors also in $\log t$.
\end{lemma}

\begin{proof}
First of all, to accommodate the call in line \ref{lin:boost} we note that any $(a,\zeta)$-effective algorithm with success probability $9/10$ can be turned  to an $(a,\zeta)$-effective algorithm with success probability $1-\delta$ by running in parallel $\Theta(\log (1/\delta))$ independent copies and stopping when a $7/10$ fraction of them has halted. Then at least half of them will have returned the same answer, and we can return this. Fix $x \in \mathcal{S}(X,t)$ and let $I \subseteq X$ such that $\Sigma(I) = x$. Since every element of $X$ is at least $u/(\beta \log^3 t)$, it should be the case that $|I| \leq \beta \log^3 t$. For a fixed $r \in [R]$, the probability that all elements in $I$ are colored with a different color is at least $\frac{1}{2}$ by standard arguments. Conditioning on the latter event happening, the computation in Lines~\ref{lin:large_sumset1},~\ref{lin:large_sumset2} will contain $x$ with probability, by the same argument as the one appearing in~\cite[Lemma 3.1]{Bring17}. In particular, $x \in X^{(r,1)} + \ldots + X^{(r,B)}$, by the fact that every element of $I$ belongs to a distinct set $X^{(r,b)}$, and those sets contain $0 \in X^{(r,b)}$ for all $b$.

All calls to prefix-restricted sumset computation will succeed with probability $\delta/2$ by a union bound. Repeating $R= \Theta(\log(u/\delta))$ times ensures that the coloring in line \ref{lin:color_coding_large} will color every element of $I$ with a different color, i.e. for some $r$ it hold that $x \in X^{(r,1)} + \ldots + X^{(r,B)}$ except with probability $\left(\frac{1}{2}\right)^R = \delta/ (2u )$. In turn, this implies that $x$ will be found with probability $ 1- \frac{\delta}{2u}$. Taking a union-bound over the at most $u$ possible values of $x$ we obtain that $O$ shall contain every $x \in S(X,u)$ with probability $1- \left(\frac{\delta}{2} + \frac{\delta}{2}\right) = 1-\delta$. The running time of the algorithm follows by the fact that each computation involved is of the form $(A + B) \cap [u]$ with $A,B\in \mathcal{S}(X,u)$, and there are $R \cdot B = O( \log^6t \cdot \log(u/\delta))$ such computations. This finishes the proof.

\end{proof}

\subsection{General Algorithm}\label{sec:main_subset_sum}

\begin{algorithm}[!t]\caption{}\label{alg:general_algo}
\begin{algorithmic}[1]
\Procedure{\textsc{SubsetSumReduction}}{$X,u,t$}   \Comment{ $u \leq t$}
\State $X^{(S)} \leftarrow X \cap \left[\frac{u}{\beta \log^3 t}\right] $	\Comment{Partition $X$ to small and large element}
\State $ O \leftarrow \textsc{SubsetSumforLargeElements}(X\setminus X^{(S)},u,t,1/\poly(u))$\label{lin:large_call}
\State $X^{(1)} \leftarrow$ Sample $X^{(S)}$ at rate $1/2$  \label{line:splitting}  
\State $X^{(2)} \leftarrow X^{(S)} \setminus X^{(1)}$
\State $\epsilon \leftarrow \frac{1}{\log t}$
\State $O_1 \leftarrow \textsc{SubsetSumReduction}\left(X^{(1)},(1+\epsilon)\frac{u}{2},t\right)$ \label{lin:general_lin1}
\State $O_2 \leftarrow \textsc{SubsetSumReduction}\left(X^{(2)},(1+\epsilon)\frac{u}{2},t\right)$ \label{lin:general_lin2}
 \State Return $ (O + (O_1 + O_2) \cap [u]) \cap [u]$ \Comment{Oracle call with failure probability $1/\poly(t)$}
\EndProcedure
\end{algorithmic}
\end{algorithm}

The reduction is presented in Algorithm~\ref{alg:general_algo}. As mentioned in the overview, the algorithm partitions the set $X$ to small elements (set $X^{(S)}$) and large elements (set $X\setminus X^{(S)}$).  The large elements are handled by the algorithm presented in the previous subsection, Line~\ref{lin:large_call}. Next, $X^{(S)}$ is split to $X^{(1)},X^{(2)}$ randomly in Line~\ref{line:splitting}, and recursion takes places in each part with the appropriate change of the target (Lines~\ref{lin:general_lin1},\ref{lin:general_lin2}).

\subsection{A Number-Theoretic Lemma for the decrease of Subset Sums}

The crux of the analysis is the following technical lemma, which postulates that the number of subsets sums of a set decreases in a suitable way when halving the set. This will allows us to control the running time of Algorithm~\ref{alg:general_algo}. Note that this holds \emph{for any} partitioning, not only for a random one.

\begin{lemma} [Number of subset sums decreases appropriately]\label{lem:splitting_lemma}
 Let $Z \subseteq [u]$ be partitioned into $Z = Z^{(1)} \cup Z^{(2)}$. Set $\mu := \mathrm{max}(Z) / u$ and assume $\mu \leq \frac{1}{16}$. Then for any $0 \le \epsilon \leq \frac{1}{4}$ we have
\begin{align*}
\left|\mathcal{S}\left(Z^{(1)},(1+\epsilon)\frac{u}{2}\right)\right| + \left|\mathcal{S}\left(Z^{(2)},(1+\epsilon)\frac{u}{2}\right)\right| \leq \frac{|\mathcal{S}(Z,u)| +1}{ 1 - 2 \epsilon - 4 \mu}
\end{align*}
\end{lemma}

\begin{proof}
We denote 
\begin{align*}
A &:= \mathcal{S}\left(Z^{(1)},(1+\epsilon)\frac{u}{2}\right),\\
B &:= \mathcal{S}\left(Z^{(2)},(1+\epsilon)\frac{u}{2}\right), \\
C &:= \mathcal{S}(Z,u),\\
A' &:= A \cap\left((1-\epsilon-2\mu)\frac{u}{2}, (1+\epsilon)\frac{u}{2}\right].
\end{align*}
With this notation, our goal is to show
\begin{align} \label{eq:mytoshow}
  |C| \ge (|A|+|B|)(1-2\epsilon-4\mu) - 1.
\end{align}



By symmetry, without loss of generality we can assume $\max(A) \le \max(B)$. 

It remains to consider the case $\max(A), \max(B) > (1-\epsilon) \frac u2$.

 The statement is trivial if $\mathrm{max}(A) \leq (1-\epsilon)\frac{u}{2}$ (or $\mathrm{max}(B) \leq (1-\epsilon)\frac{u}{2})$, since then for $x \in B \setminus \{ 0 \}$ we can map $I_Y(x)$ to $I_Y(x) \cup W$, yielding the following:  $\Sigma(I_Y(x) \cup W) = x + \Sigma(W) > \mathrm{max}(A)$, obtaining $|B|-1$ distinct subset sums above $\mathrm{max}(A)$, from which we conclude that $|C| \geq |A| + |B| - 1$.

Let us now assume that this is not the case, and thus $\mathrm{max}(A) \geq (1-\epsilon)\frac{u}{2}$ and $\mathrm{max}(B) \geq (1-\epsilon)\frac{u}{2}$. We shall need the following two claims.

\begin{claim} $ |C| \geq |A| + |B| - |A'|-1$ \end{claim} \label{claim:claim_1}
\begin{proof}
Note that
\begin{align*}
\left|C \cap \left[0,(1-\epsilon-2\mu)\frac{u}{2} \right]\right| \geq \left|A \cap \left[0,(1-\epsilon-2\mu)\frac{u}{2} \right ]\right|
= |A|-|A'|
\end{align*}

Moreover, since all items are bounded by $\mu u$, we can choose $P \subseteq W$ such that 
\[	\Sigma(P) \in \left[(1-\epsilon-2\mu)\frac{u}{2}, (1-\epsilon)\frac{u}{2}\right].\] 
To see that, let any ordering of elements of $W$, initialize $P$ to the empty set, and start adding elements to it one by one; clearly at some point $\Sigma(P)$ will fall inside the aforementioned interval. Now, for every $x \in B\setminus \{0 \}$, map $I_Y(x)$ to $I_Y(x) \cup P$ to obtain number $\Sigma(I_Y(x) \cup P) = x + \Sigma(P)$; this yields $|B|-1$ different sums, all in the interval $\left[(1-\epsilon-2\mu)\frac{u}{2} +1, u\right]$, and thus disjoint from the numbers in $A$ counted above. Thus, we obtain at least $(|A|-|A'|)+(|B|-1)$ different numbers in $C$, yielding the proof of the claim.

\end{proof}

\begin{claim} $|C| \geq |A| + |A'|  \left(\frac{1}{2\epsilon+4\mu} - 1\right).$  \end{claim} \label{claim:claim_2}

\begin{proof}
In order to prove the desired lower bound, we shall look at two disjoint intervals  $\left[0,(1+\epsilon)\frac{u}{2}\right]$, and $\left[(1+\epsilon)\frac{u}{2}+1,t\right]$.

In the interval $\left[0,(1+\epsilon)\frac{u}{2}\right]$ we shall simply count the numbers in $A$, $\left|A\cap\left[0,(1+\epsilon)\frac{u}{2}\right] \right|\leq |A|$.  

For the other interval we argue as follows. There exists a sequence of sets 
\[	P_{i_0} \subseteq P_{i_0+1} \subseteq P_{i_0+2} \subseteq \ldots \subseteq Y\]
 satisfying

\[	\Sigma(P_i) \in \left[i(\epsilon+2\mu) u +1 , i (\epsilon+2\mu) u + \mu u)	\right],	\]

for all $i \geq 1$ such that
\begin{align}
i (\epsilon+2\mu) u + \mu u \leq \mathrm{max}(B).
\end{align}

Note that $i_0$ is the smallest non-zero $i$ such that the above inequality holds. Call such $i$ good.

To see the existence of such a sequence, intialize $P$ to the empty set and starting adding elements of $Y$ one by one. Since every element is at most $\mu u$ and the $ \left[i(\epsilon+2\mu) u +1 , i (\epsilon+2\mu) u + \mu u)	\right]$ is of length $\mu$ we obtain that existence of such a sequence.

Since $ \mathrm{max}(B) \geq (1-\epsilon)\frac{u}{2}$  all $i$ smaller than 
\begin{align*}
\frac{(1-\epsilon)\frac{u}{2} - \mu u}{ (\epsilon+2\mu)u} = \frac{ 1-\epsilon - 2\mu}{2(\epsilon+2\mu)}\\
= \frac{1}{2\epsilon+4\mu} - \frac{1}{2}.
\end{align*}

are good.

For any $ x \in A' \subseteq \mathcal{S}(X,(1+\epsilon)\frac{u}{2})$ map $I_{W}(x)$ to $I_{W}(x) \cup P_i$, to obtain $|A'|$ different numbers in the interval
\[ \left[i(\epsilon+2\mu) u+1, i (\epsilon+2\mu) u + \mu u)\right] +  \left[(1-\epsilon-2\mu)\frac{u}{2}, (1+\epsilon)\frac{u}{2}\right]  = \] 
\[ \left[i(\epsilon+2\mu)u + (1-\epsilon-2\mu)\frac{u}{2} +1 ,
(i+1)(\epsilon+2\mu)u + (1-\epsilon-2\mu)\frac{u}{2} \right]. \]
The collection of those intervals across all $i$ are pairwise disjoint as well as disjoint from the initial interval $\left[0,(1+\epsilon)\frac{u}{2}\right]$. 

In order for all generated sums to be at most $u$, we also need
$(i+1)(\epsilon+2\mu) u + (1-\epsilon - 2\mu)\frac{u}{2} \leq u$, which boils down to 
$i +1 \leq \frac{ 1+ \epsilon + 2\mu}{2\epsilon+4\mu} = \frac{1}{2} + \frac{1}{2\epsilon+4\mu}$. Since $i$ is an integer, we have at least
$\frac{1}{2\epsilon+4\mu}-1$  valid $i$'s.
Hence, we obtain
\[	|C| \geq |A| + |A'|  \left(\frac{1}{2\epsilon+4\mu} - 1\right).	\]

\end{proof} 
To finish the proof of the lemma we combine the two claims, by considering two cases:
\begin{itemize}
\item Case 1: $|A'| \leq (2\epsilon+4\mu) |B|$.

Then Claim \ref{claim:claim_1} yields
	\[	|C| \geq |A| + |B| - |A'| -1  \geq |A| + |B|(1-2\epsilon-4\mu) -1 \geq (|A| + |B|) (1-2\epsilon-4\mu) -1. \]
\item Case 2: $|A'| \geq (2\epsilon+4\mu) \cdot |B|$

Then Claim \ref{claim:claim_2} yields
\begin{align*}
|C| \geq |A| + |A|'\left(\frac{1}{2\epsilon + 4\mu} \right)\geq \\
|A| + |B|(2\epsilon + 4\mu) \left( \frac{1}{2\epsilon + 4\mu}  \right) \geq \\
 (|A| + |B|)(1-2\epsilon-4\mu).
\end{align*}
\end{itemize}
\end{proof}

\subsection{Putting Everything Together}

We finish the reduction and the proofs of Theorems~\ref{thm:add_relaxed_upper_bound} and~\ref{thm:subsetsum_result}. The algorithms is one call to $\textsc{SubsetSumReduction}(S,t,t)$.

\paragraph{Proof of correctness.}

We show that $\textsc{SubsetSumReduction}(X,t,t)$ returns $\mathcal{S}(X,t)$ with constant probability. This will require that all recursive calls succeed. In particular, for every call in the 
\begin{claim}\label{claim:bernstein}
Consider the execution of $\textsc{SubsetSumReduction}(X,u,t)$. Fix $x \in \mathcal{S}(X,u)$ and set $I = X^{(S)} \cap I_{X}(x)$, i.e. the ``part'' of the representation $x$ which is formed by small elements.
It holds that
\[\mathbb{P} \left\{ \Sigma(I \cap X^{(1)} )  \notin \left[(1+\epsilon)\frac{u}{2}\right] \right\} \leq \frac{1}{\poly(t)}.\] In words, the sum of elements of $I$ which belong to $X^{(1)}$ will be at most $(1+\epsilon)\frac{u}{2}$ with high probability. The analogous statement holds with $X^{(1)}$ replaced with $X^{(2)}$.
\end{claim}

\begin{proof}

This claim easily follows by concentration bounds for bounded random variables. We shall use Bernstein's inequality which postulates that for a collection $\mathcal{C}$ of random variables $\{Z_e\}_{e \in \mathcal{C}}$ such that all $Z_e \in [0,K]$, it holds that

	\[	\mathbb{P} \left\{ \left|\sum_e Z_e - \mathbb{E}\sum_e Z_e\right| \geq \lambda \right\} \leq e^{-\frac{c\lambda}{K}} + e^{-\frac{c \lambda^2}{\sigma^2}},		\]

where $\sigma^2 = \sum_e \mathbb{E} \left\{ (Z_e - \mathbb{E}Z_e)^2 \right\}, \lambda \geq 0$ and $c$ is some absolute constant.

We apply the inequality the collection $C:= I$ of independent random variables $\{Z_e\}_{e \in I}=\{ e \cdot \mathbbm{1}_{X^{(1)}}(e) \}_{e \in I}$ and $\lambda = \epsilon \frac{u}{2}$. In words, $Z_e$ is $e$ with probability $1/2$, and $0$ otherwise. We have 
\begin{enumerate}
\item $\mathbb{E}\left[ \sum_e Z_e \right] \leq \frac{x}{2} \leq \frac{u}{2}$.
\item $K = \frac{u}{ \beta \cdot \log^3 t} $ by definition of $X^{(S)}$, and 
\item $\sigma^2 \leq \frac{Kx}{2}$ since  
\[ \sigma^2 \leq \sum_{e \in I} e^2 \cdot \mathbb{E} \left\{ \mathbbm{1}_{X^{(1)}}(e)\right\}  \leq  K \sum_{e \in I} e \cdot \mathbb{E} \left\{ \mathbbm{1}_{X^{(1)}}(e)\right\} = \frac{1}{2} K \sum_{e \in I}  \frac{e}{2} \leq \frac{Kx}{2}.\]
\end{enumerate}

We thus obtain 
	\[	\mathbb{P} \left\{ \Sigma( I \cap X^{(1)})  \geq \epsilon \frac{u}{2} \right\} \leq e^{-\frac{c \beta \cdot \epsilon \log^3 t }{4} }  + e^{-\frac{c \beta \cdot \epsilon^2 u^2 \cdot \log^3 t }{2 u x} } \leq 1/\poly(t),	\]

as long as $\beta$ is sufficiently large compared to $c$. This finishes the proof of the claim.

\end{proof}

Equipped with the above claim, we can now prove correctness of the reduction. Fix $x \in \mathcal{S}(X,t)$ and let $I = I_{X}(x)$. Let also $I_0 = I \cap (X\setminus X^{(S)}), I_1 = I \cap X^{(1)}, I_2 = I \cap X^{(2)}$. Set also $ x = y+ z+w$, where $y = \Sigma(I_0) , z =\Sigma(I_1), w = \Sigma(I_2)$. It holds that $y$ will be returned by Algorithm~\ref{alg:general_algo} with probability $1-1/\poly(t)$ due to the call in Line~\ref{lin:large_call}. If the conclusion of Claim \ref{claim:bernstein} holds for $I$, then this means that both $\Sigma(I \cap X^{(1)})$ and $\Sigma(I \cap X^{(2)})$ are at most $(1+\epsilon)\frac{u}{2}$ and hence $z \in \mathcal{S}(X^{(1)}, (1+\epsilon)\frac{u}{2})$, $w\in \mathcal{S}(X^{(2)}, (1+\epsilon)\frac{u}{2})$. Thus, if the recursive calls in Lines~\ref{lin:general_lin1} and~\ref{lin:general_lin2} as well as the call~\ref{lin:large_call} succeed, then $x = y+z+w$ will be inserted to the output. This means that the correctness of the algorithm is guaranteed on the conclusion of Claim~\ref{claim:bernstein} holding in all recursive calls and on every call to \textsc{SubsetSumforLargeElements} being correct. Since in each recursive call $t$ does not change and remains the same, each recursive splitting succeeds with probability $1-1/\poly(t)$, while every call to $\textsc{SubsetSumforLargeElements}$ succeeds also with probability $1-1/\poly(t)$. This allows for a union-bound over all splitting and all calls to $\textsc{SubsetSumforLargeElements}$.

\paragraph{Proof of Desired Running Time.}

Due to the splitting lemma \ref{lem:splitting_lemma}, the fact that $\epsilon = \frac{1}{\log t } \leq \frac{1}{\log u} \leq \frac{1}{4}$ at all times, and a straightforward induction, we have that the total output size for all problems in the $\ell$th level of the recursion tree during the execution of  $\textsc{SubsetSumReduction}(S,u,t)$ is upper bounded by
	\[	\frac{|\mathcal{S}(X,t)|}{(1-2\epsilon - 2\mu)^\ell} = \tilde{O}( |\mathcal{S}(X,t)| ),	\]


since $\ell \leq \log n \leq \log t$ and $\epsilon = \frac{1}{\log t}, \mu = \frac{1}{\beta \log^3 t}$. Thus, over all recursion levels the total output size is still $\tilde{O}( |\mathcal{S}(X,t)| )$.   This shows that if we had the ideal $(0,0)$-effective algorithm, we would obtain a \textsc{SubsetSum} running in time $\tilde{O}(\mathcal{S}(X,t))$. A similar analysis yields the desired reduction when we plug in a $(\alpha,\zeta)$-effective algorithm.

\paragraph{Obtaining Theorem~\ref{thm:subsetsum_result}.} We plug in the $\tilde{O}(\out^{4/3})$-time algorithm for prefix-restricted sumset computation (Theorem~\ref{thm:technical_core}). 
\paragraph{Obtaining Theorem~\ref{thm:add_relaxed_upper_bound}.} We plug in the algorithm guaranteed by the first part of Theorem~\ref{thm:add_relaxed_upper_bound} and Observation~\ref{obs:sumset_from_covering}.

%% file: conclusion.tex
\section{Acknowledgements}

We are grateful to Shachar Lovett for the resolution of an Additive Combinatorics question in an early stage of this work, which gave the core idea for Theorem \ref{thm:covering_lb}, and for allowing us to include his construction in this paper.

\section{Conclusion and Future Work}

We initiated a line of research which strives for a \textsc{SubsetSum} algorithm that computes the set $\mathcal{S}(X,t)$, consisting of all subsets sums of $X$ below $t$, in near-linear output-sensitive time $\tOh(|\mathcal{S}(X,t)|)$. Our approach lead us to studying a new type of convolution problem: In top-$k$-convolution the task is to compute the lowest $k$ monomials in the product of two sparse polynomials. Many open problems are spawned by our work; here we present questions that are of particular interest to us.

\begin{question}
Understand our notion of \emph{covering} for prefix-restricted sumset computation, either by constructing a better covering or by proving a higher lower bound. Specifically, for non-rectangular coverings so far we have no superlinear lower bounds.
\end{question}

\begin{question}
Design any non-trivial algorithm that is not based on coverings, and thus exploits the additive structure in a different way.
\end{question}

\begin{question}
Are covering algorithms universal? More precisely, can we transform any algorithm into a covering algorithm, with a reasonable blow-up in the running time?
\end{question}

\begin{question}
So far we have no algorithm that always beats Bellman's algorithm with running time $\Oh(n \cdot |\mathcal{S}(X,t)|)$. Specifically, can we solve \textsc{SubsetSum} in time $O(n^{1-\varepsilon} \cdot |\mathcal{S}(X,t)|)$ for any $\varepsilon > 0$? 
A possible approach to this question is to design faster algorithms for prefix-restricted sumset computation in the special situation where $A = \mathcal{S}(X^{(1)},t), B = \mathcal{S}(X^{(2)},t)$.
Do these sets offer exploitable structure, e.g., giving rise to more sophisticated sumset estimates? 
\end{question}